\documentclass[preprint]{elsarticle}
\usepackage{textcomp}
\usepackage{amsmath}
\usepackage{amsthm}
\usepackage[utf8]{inputenc}
\usepackage{caption,subcaption}
\usepackage{mathrsfs}
\usepackage{xfrac}

\usepackage[usenames]{color}
\usepackage[noend]{algorithmic}
\usepackage{algorithm}
\usepackage{stmaryrd}
\usepackage{multicol}
\usepackage{pictexwd}
\usepackage{xspace}
\usepackage{balance}
\usepackage{amssymb}
\usepackage{url}

\def\e#1{\emph{#1}}
\newcommand{\partitle}[1]{\vskip 0.5em\par\noindent{\textbf{#1.}\,\,}}

\def\P{\mathsf{P}}
\def\Q{\mathcal{Q}}
\def\V{\mathcal{V}}

\def\MO{\mathcal{P}}

\def\asn{\mathbin{{:}{=}}}

\def\nodes{\mathsf{V}}
\def\edges{\mathsf{E}}
\def\set#1{\mathord{\{#1\}}}
\def\sm{\setminus}
\def\minsep{\mathit{MinSep}}
\def\clqminsep{\mathit{ClqMinSep}}
\def\mintri{\mathit{MinTri}}
\def\maxcl{\mathit{MaxClq}}
\def\maxind{\mathit{MaxInd}}

\def\crosses{\mathbin{\natural}}
\def\bags{\mathrm{bags}}
\def\equivb{\mathrel{\equiv_{\mathsf{b}}}}

\def\sepset{\varphi}

\newcommand{\eat}[1]{}

\newtheorem{theorem}{Theorem}[section]
\newtheorem{proposition}[theorem]{Proposition}
\newtheorem{corollary}[theorem]{Corollary}
\newtheorem{lemma}[theorem]{Lemma}
\newtheorem{definition}{Definition}
\newenvironment{citedtheorem}[1]
{\begin{theorem}{\it\e{(#1)}}\,\,}
{\end{theorem}}

\newtheorem{examplethm}[theorem]{Example}

\newenvironment{repeatresult}[2]
{\vskip0.5em\par\textsc{#1} #2.\em}
{\vskip1em}

\def\G{\mathcal{G}}
\def\Av{A_{\mathsf{V}}}
\def\Ae{A_{\mathsf{E}}}

\newenvironment{qproof}
{\begin{proof}}
{\end{proof}}

\newcommand{\algname}[1]{{\sf #1}}
\def\myrulewidth{3.20in}
\def\therule{\rule{\myrulewidth}{0.2pt}}

\newenvironment{algseries}[2]
{\centering\begin{figure}[#1]\begin{center}\def\thecaption{\caption{#2}}
\begin{tabular}{p{\myrulewidth}}\therule\end{tabular}\vskip0.2em}
{\thecaption\end{center}\end{figure}}

\newenvironment{insidealg}[2]
{\normalsize\begin{insidecode}{#1}{#2}{Algorithm}}
{\end{insidecode}}

\newenvironment{insidecode}[3]
{
\begin{tabular}{p{\myrulewidth}}
\multicolumn{1}{c}{\rule{0mm}{3mm}{\bf #3} $\algname{#1}(\mbox{#2})$\vspace{-0.6em}}\\
\therule\vskip-0.8em\therule
\vspace{-1em}
\begin{algorithmic}[1]}
{\end{algorithmic}
\vskip-0.3em\therule
\end{tabular}}

\def\iter{\mathrm{iterator}}
\usepackage{comment}
\def\sat{\algname{saturate}}
\def\trihrs{\algname{Triangulate}}

\def\MSCM{\algname{MCS\_M}}
\def\LBTRI{\algname{LB\_TRIANG}}
\def\Extend{\algname{Extend}}

\def\sepsys{\mathsf{MSGraph}}
\def\ms{^{\mathsf{ms}}}

\def\time{\mathrm{time}}
\def\answers{\mathrm{out}}

\newcommand{\eqdef}{\overset{\mathrm{def}}{=\joinrel=}}

\newcommand{\SETH}{\mathsf{SETH}}
\newcommand{\SAT}{\mathsf{SAT}}
\newcommand{\kSAT}{k\text{-}\SAT}
\newcommand{\calG}{\mathcal{G}}

\newcommand{\var}{\mathsf{var}}
\newcommand{\poly}{\mathsf{poly}}
\newcommand{\true}{\mathsf{true}}
\newcommand{\false}{\mathsf{false}}

\def\addedb#1{{#1}}
\def\addedn#1{{\addedb{#1}}}

\begin{document}

\begin{frontmatter}

\title{Efficiently Enumerating Minimal Triangulations}

\author[technion]{Nofar Carmeli}
\ead{snofca@cs.technion.ac.il}

\author[technion]{Batya Kenig}
\ead{batyak@cs.technion.ac.il}

\author[technion]{Benny Kimelfeld}
\ead{bennyk@cs.technion.ac.il}

\author[tuwien]{Markus Kr\"oll}
\ead{kroell@dbai.tuwien.ac.at}

\address[technion]{Technion, Haifa, Israel}
\address[tuwien]{TU Wien, Vienna, Austria}

\begin{abstract}
  We present an algorithm that enumerates all the minimal
  triangulations of a graph in incremental polynomial time.
  Consequently, we get an algorithm for enumerating all the proper
  tree decompositions, in incremental polynomial time, where
  ``proper'' means that the tree decomposition cannot be improved by
  removing or splitting a bag.  The algorithm can incorporate any
  method for (ordinary, single result) triangulation or tree
  decomposition, and can serve as an anytime algorithm to improve such
  a method. We describe an extensive experimental study of an
  implementation on real data from different fields. Our experiments
  show that the algorithm improves upon central quality measures over
  the underlying tree decompositions, and is able to produce a large
  number of high-quality decompositions. 
\end{abstract} 

\end{frontmatter}

\section{Introduction}
Many intractable computational problems on graphs admit tractable
algorithms when applied to trees or forests. A \e{tree decomposition}
extracts a tree structure from a graph by grouping nodes into \e{bags}
(each treated as a single node). The corresponding operation on
hypergraphs is that of a \e{generalized hypertree
  decomposition}~\cite{DBLP:journals/jair/GottlobGS05} that consists
of a tree decomposition of the \e{primal} graph (which has the same
set of nodes, and an edge between every two hyperedge neighbors), and
an assignment of hyperedge labels (edge covers) to the tree
nodes~\cite{DBLP:conf/wg/GottlobGMSS05}.  Tree decompositions and
generalized hypertree decompositions have a plethora of applications,
including optimization of (multi)join queries in
databases~\cite{DBLP:conf/sigmod/TuR15,DBLP:journals/jair/GottlobGS05,DBLP:journals/corr/KalinskyEK16},
containment of database queries, constraint satisfaction
problems~\cite{DBLP:journals/jcss/KolaitisV00}, prediction of RNA
secondary structure~\cite{DBLP:conf/wabi/ZhaoMC06}, computation of
Nash equilibria in games~\cite{DBLP:journals/jair/GottlobGS05},
inference in probabilistic graphical models~\cite{LauSpi-JRS88}, and
weighted model counting~\cite{DBLP:conf/sum/KenigG15}.

Past research has focused on obtaining a ``good'' tree decompositions,
where goodness is typically defined as having low \e{tree width}~\cite{ROBERTSON198449}---the maximal
cardinality of a bag (minus one). Nevertheless, finding a tree
decomposition of the minimal tree width is
NP-hard~\cite{arnborg1987complexity}, as is the case for other common
measures of goodness for tree decompositions such as
\e{fill}~\cite{doi:10.1137/0602010}, and in the case of hypergraphs
\e{hypertree width}~\cite{GOTTLOB2002579}, \e{generalized
  hypertree width}~\cite{DBLP:journals/jacm/GottlobMS09}, and
\e{fractional hypertree width}~\cite{DBLP:journals/talg/Marx10}.
Therefore, heuristic algorithms are often
applied~\cite{Berry:2002:MCS:647683.732496,BERRY200633}.  The
different measures of goodness are motivated by the fact that the
needs of different applications are often different from (though
related to) the \eat{mere}width. Additional examples are the
complexity of weighted model counting, induced by a parameter
associated with the ``CNF-tree'' of the formula\eat{the ``CNF-tree''
  induced by the decomposition in weighted model
  counting}~\cite{DBLP:conf/sum/KenigG15,DBLP:conf/wg/GottlobGMSS05},
and the effectiveness of \e{adhesions} (parent-child intersection) for
caching in terms of dimension and
skew~\cite{DBLP:journals/corr/KalinskyEK16}. In fact, Kalinsky et
al.~\cite{DBLP:journals/corr/KalinskyEK16} have illustrated how, in
real-life scenarios, isomorphic tree decompositions of a minimal
width may result in orders-of-magnitude difference in (join)
performance.

The common approach is to devise a decomposition algorithm (exact,
approximate or heuristic) to capture the desired measure of goodness
per application. However, this is a nontrivial challenge that (to the
least) requires high expertise in algorithms and tree
decompositions. We propose an alternative approach---produce a large
number of different tree decompositions, using a baseline
decomposition method, and allow the application at hand to choose the
best according to its internal measure function.  Our approach brings
together results and techniques from the areas of
\e{chordal graphs} and \e{enumeration theory} in order to establish a
practical tool for enhancing decomposition algorithms and, by
implication, the performance of various inference and optimization
algorithms. Specifically, we explore the task of \e{enumerating} all
(or a subset of) the tree decompositions. Such algorithms have been
proposed in the past for small graphs (representing database queries),
without complexity guarantees~\cite{DBLP:conf/sigmod/TuR15}.  Our main
result is an enumeration algorithm that runs in \e{incremental polynomial time}~\cite{DBLP:journals/ipl/JohnsonP88}, that is,
the time between producing the $N$th result and the $(N+1)$st result is polynomial in $N$ and in the size of the input.

We first need to define which tree decompositions should be
enumerated, as many of them are effectively useless. For example, if
we take a graph that is already a tree, we do not wish to enumerate
the tree decompositions that group nodes with no reason; in fact, the
tree itself is the only reasonable decomposition in this
case. Therefore, we consider only tree decompositions that cannot be
``improved'' by removing or splitting a bag, and we call such tree
decompositions \e{proper}. We prove that the proper tree
decompositions are in a bijective (and efficiently computable)
correspondence to the \e{minimal triangulations} of the graph at hand,
defined a follows. A \e{triangulation} of a graph $g$ is a graph $g'$
that is obtained from $g$ by adding edges so that $g'$ is \e{chordal},
that is, $g'$ does not have any induced simple cycle of more than
three nodes.  A triangulation is \e{minimal} if no triangulation can
be obtained using only a strict subset of the added edges.  

So, the problem is reduced to the task of enumerating all of the
minimal triangulations of a graph. In this manuscript we devise an
algorithm for performing this task in incremental polynomial time.
Our approach is as follows. Parra and
Scheffler~\cite{DBLP:journals/dam/ParraS97} have shown that there is a
one-to-one correspondence between the minimal triangulations of a
graph $g$ and the maximal independent sets of a special graph
$\G$. The nodes of $\G$ are the so called \e{minimal separators of
  $g$}, and the edges are between \e{crossing} minimal
separators. (Precise definitions are in
Section~\ref{sec:preliminaries}.)  So, enumerating the minimal
triangulations of a graph boils down to enumerating these maximal
sets. It is well known that all the maximal independent sets of a
graph can be enumerated with polynomial
delay~\cite{DBLP:journals/ipl/JohnsonP88,DBLP:journals/jcss/CohenKS08}.
However, this is insufficient for us, since the graph $\G$ is not
given as input, and in fact, its number of nodes can be exponential in
the size of the original graph $g$. Therefore, we cannot construct
this graph ahead of time, and cannot directly use existing algorithms
to establish incremental polynomial time.

We address this problem by defining an abstraction of the graph $\G$
of minimal separators by means of a \e{Succinct Graph Representation}
(SGR), which is represented compactly by two algorithms: one for
enumerating the nodes and one for testing whether a given pair of
nodes forms an edge.  In particular, we can access the nodes of $\G$
through a polynomial-delay iterator, due to a result by Berry et
al.~\cite{conf/wg/BerryBC99} (who show how to enumerate the minimal
separators of a graph).  Applying previous results, we prove that the
SGR for the minimal separator graph (i.e., $\G$) meets certain
tractability conditions termed \e{tractable expansion}, which enable
the enumeration of its maximal independent sets (i.e., $g$'s minimal
triangulations) in incremental polynomial time in the size of the
representation (which can be logarithmic in the size of the graph
itself).

In summary, we reduce the problem of enumerating the proper tree
decompositions to that of enumerating the minimal triangulations,
which we reduce to the problem of enumerating the maximal independent
sets of an SGR with tractability properties, and we devise an
algorithm for the latter task. An important feature of the algorithm
is that it can incorporate any black-box procedure for expanding a
given independent set into a maximal one. When applied to enumerating
the proper tree decompositions, such a procedure can be any
off-the-shelf algorithm for minimal triangulation or tree
decomposition (e.g., Maximum Cardinality
Search~\cite{Berry:2002:MCS:647683.732496} and
LB-Triang~\cite{BERRY200633}). However, our algorithm executes this
procedure on different versions of the original graph, each time with
some new edges added. Hence, our algorithm has the potential of using
a high-quality decomposition algorithm for producing \e{many}
high-quality decompositions, enabling the user to choose the best one
generated according to the specific measures of her use case (may it
be width or anything else).

After establishing our algorithm, we describe an experimental study
where we have tested the ability of the algorithm to utilize the
aforementioned triangulation algorithms.  The experimental study
covers graphs of a wide range of domains (where tree decomposition is
needed for efficient analysis): join queries (from the TPC-H
collection), Bayesian networks, Markov Random Fields, grids, and
random graphs. We tested the execution time (delay) of the algorithm,
its ability to reduce the width or fill (number of edges added to
establish chordality), and the number of decompositions of quality
(width/fill) the same or better than that of the original
off-the-shelf algorithm. The results show that, indeed, our algorithm
can effectively enhance the quality of the corresponding decomposition
algorithm.

A short version of this manuscript was published in the 2017 Symposium on Principles of Database Systems (PODS). The current version includes all proofs omitted from the short version. In addition, we have added Section~\ref{sec:poly} and Section~\ref{sec:sgrpdelay}, where we discuss the possibility of reaching a polynomial delay algorithm solving our problem, and prove that the approach taken here cannot be used to reach such a time bound.

The rest of the manuscript is organized as follows. In
Section~\ref{sec:preliminaries} we give preliminary definitions and
notation, and recall basic results from the literature.  The SGR
framework is presented in Section~\ref{sec:SGR}, along with the
enumeration algorithm for maximal independent sets. In
Section~\ref{sec:enumMT} we prove that the graph of minimal separating
sets satisfies the tractability requirements needed for the SGR
enumeration algorithm, and thereby establish an algorithm for
enumerating the minimal triangulations. We show how this algorithm can
enumerate the proper tree decompositions in
Section~\ref{sec:proper}. Then, the experimental study is presented in
Section~\ref{sec:exper}, and we conclude in
Section~\ref{sec:conclusions}.

\section{Preliminaries}\label{sec:preliminaries}
In this section we give some basic notation and terminology that we
use throughout the paper. In addition, we recall some basic theory
that we need in this paper.

\subsection{Graphs}
The graphs in this work are undirected. For a graph $g$, the set of
nodes is denoted by $\nodes(g)$, and the set of edges (pairs
$\set{u,v}$ of distinct nodes) is denoted by $\edges(g)$.  Let $U$ be
a set of nodes of a graph $g$.  We denote by $g_{|U}$ the subgraph of
$g$ \e{induced} by $U$; that is, $\nodes(g_{|U})=U$ and
$\edges(g_{|U})=\set{\set{u,v}\in\edges(g)\mid \set{u,v}\subseteq U}$.
We denote by $g\sm U$ the graph obtained from $g$ by removing all the
nodes in $U$ (along with their incident edges), that is, the graph
$g_{|\nodes(g)\sm U}$. 

Let $g$ be a graph and $U$ a set of nodes of $g$. We say that $U$ is
an \e{independent set} if it does not contain both endpoints of any
edge, and it is a \e{maximal independent set} if it is an independent
set and it is not strictly contained in any other independent set. We
denote by $\maxind(g)$ the set of all the maximal independent sets of
$g$.  We say that $U$ is a \e{clique} (\e{of $g$}) if every two nodes of $U$ are
connected by an edge, and it is a \e{maximal clique} (\e{of $g$}) if it is a clique
that is not strictly contained in any other clique.  We denote by
$\maxcl(g)$ the set of all the maximal cliques of $g$. The operation
of \e{saturating} $U$ (\e{in $g$}) is that of connecting every
non-adjacent pair of nodes in $U$ by a new edge. Hence, if $h$ is
obtained from $g$ by saturating $U$, then $U$ is a clique of $h$.

\subsection{Minimal Separators}
Let $g$ be a graph, and let $S$ be a subset of $\nodes(g)$.  Let $u$
and $v$ be two nodes of $g$. We say that $S$ is a
\e{$(u,v)$-separator} if $u$ and $v$ belong to distinct connected
components of $g\setminus S$.  We say that $S$ is a \e{minimal}
$(u,v)$-separator if no strict subset of $S$ is a $(u,v)$-separator.
We say that $S$ is a \e{minimal separator} if there are two nodes $u$
and $v$ such that $S$ is a minimal $(u,v)$-separator.  We denote by
$\minsep(g)$ the set of all the minimal separators of $g$.  We mention
that the number of minimal separators (i.e., $|\minsep(g)|$) may be
exponential in the number of nodes (i.e., $|\nodes(g)|$).

Let $g$ be a graph, and let $S$ and $T$ be two minimal separators of
$g$. We say that $S$ \e{crosses} $T$, in notation $S\crosses_g T$, if
there are nodes $u$ and $v$ in $T$ such that $S$ is a
$(u,v)$-separator.  If $g$ is clear from the context, we may omit it
and write simply $S\crosses T$.  It is known that $\crosses$ is a
symmetric relation: if $S$ crosses $T$ then $T$ crosses
$S$~\cite{DBLP:journals/dam/ParraS97,DBLP:journals/tcs/KloksKS97}. Hence,
if $S\crosses T$ then we may also say that $S$ and $T$ are
\e{crossing}. When $S$ and $T$ are non-crossing, then we also say that
$S$ and $T$ are \e{parallel}.

\subsection{Chordality and Triangulation}

Let $g$ be a graph. A \e{cycle} of $g$ is a path that starts and ends
with the same node. A \e{chord} of a cycle $c$ of $g$ is an edge
$e\in\edges(g)$ that connects two nodes that are non-adjacent in
$c$. We say that $g$ is \e{chordal} if every cycle of length greater
than three has a chord. Whether a given graph is chordal can be
decided in linear time~\cite{Tarjan:1984:SLA:1169.1179}.
Dirac~\cite{Dirac1961} has shown a characterization of chordal graphs
by means of their minimal separators.

\begin{citedtheorem}{Dirac~\cite{Dirac1961}} \label{thm:Dirac}
  A graph $g$ is chordal if and only if every minimal separator of $g$
  is a clique.
\end{citedtheorem}

Rose~\cite{Rose1970597} has shown that a chordal graph $g$ has fewer
minimal separators than nodes (that is, if $g$ is chordal then
$|\minsep(g)|<|\nodes(g)|$), and Kumar and
Madhavan~\cite{Kumar1998155} have shown that we can find all of these
minimal separators in linear time.

\begin{citedtheorem}{Kumar and Madhavan~\cite{Kumar1998155}}\label{thm:minsep-chordal-ptime}
  Let $g$ be a chordal graph. The set $\minsep(g)$ can be computed in linear time.
\end{citedtheorem}

A \e{triangulation} of a graph $g$ is a graph $h$ such that
$\nodes(g)=\nodes(h)$, $\edges(g)\subseteq\edges(h)$, and $h$ is
chordal. The edges in $\edges(h)\setminus\edges(g)$ are commonly
referred to as \e{fill edges}.  A \e{minimal triangulation} of $g$ is
a triangulation $h$ of $g$ with the following property: for every
graph $h'$ with $\nodes(g)=\nodes(h')$, if
$\edges(g)\subseteq\edges(h')\subsetneq\edges(h)$ then $h'$ is
non-chordal (or in other words, $h'$ is not a triangulation of
$g$). In particular, if $g$ is already chordal then $g$ is the only
minimal triangulation of itself. We denote by $\mintri(g)$ the set of
all the minimal triangulations of $g$.

\subsection{Tree Decomposition}
Let $g$ be a graph. A \e{tree decomposition} $d$ of $g$ is a pair
$(t,\beta)$, where $t$ is a tree and
$\beta:\nodes(t)\rightarrow2^{\nodes(g)}$ is a function that maps
every node of $t$ into a set of nodes of $g$, so that all of the
following hold.
\begin{itemize}
\item Nodes are covered: for every node $u\in\nodes(g)$ there is a
  node $v\in\nodes(t)$ such that $u\in\beta(v)$.
\item Edges are covered: for every edge $e\in\edges(g)$ there is a
  node $v\in\nodes(t)$ such that $e\subseteq\beta(v)$.
\item \e{Junction-tree} (or \e{running-intersection}) property: for
  all nodes $u,v,w\in\nodes(t)$, if $v$ is on the path between $u$ and
  $w$, then $\beta(v)$ contains $\beta(u)\cap\beta(w)$.
\end{itemize}

Let $g$ be a graph, and let $d=(t,\beta)$ be a tree decomposition of
$g$. For a node $v$ of $t$, the set $\beta(v)$ is called a \e{bag} of
$d$. We denote by $\bags(d)$ the set $\set{\beta(v)\mid
  v\in\nodes(t)}$, and we denote by $\sat(g,d)$ the graph obtained from
$g$ by saturating (i.e., adding an edge between every pair of nodes
in) every bag of $d$.

Jordan~\cite{Jordan} shows the following characterization of chordal
graphs by means of tree decompositions.
\begin{citedtheorem}{Jordan~\cite{Jordan}}\label{theorem:Jordan}
  A graph $g$ is chordal if and only if it has a tree decomposition
  $d$ such that every bag of $d$ is a clique of $g$.
\end{citedtheorem}

\subsection{Enumeration}

An \e{enumeration problem} $\P$ is a collection of pairs $(x,Y)$ where
$x$ is an \e{input} and $Y$ is a finite set of \e{answers} for $x$,
denoted by $\P(x)$. A \e{solver} for an enumeration problem $\P$ is an
algorithm that, when given an input $x$, produces (or \e{prints}) a
sequence of answers such that every answer in $\P(x)$ is printed
precisely once.
A solver for an enumeration problem is also referred to as an \e{enumeration algorithm}.

Johnson, Papadimitriou and
Yannakakis~\cite{DBLP:journals/ipl/JohnsonP88} introduced several
different notions of \e{efficiency} for enumeration algorithms, and we
recall these now. Let $\P$ be an enumeration problem, and let $A$ be
solver for $\P$. We say that $A$ runs in:
\begin{itemize}
\item \e{polynomial total time} if the total execution time of $A$ is
  polynomial in $(|x|+|\P(x)|)$;
\item \e{polynomial delay} if the time between printing every two
  consecutive answers is polynomial in $|x|$;
\item \e{incremental polynomial time} if, after printing $N$ answers,
  the time to print the next $(N+1)$st answer is polynomial in
  $(|x|+N)$.  \footnote{The definition of Johnson et
    al.~\cite{DBLP:journals/ipl/JohnsonP88} requires the delay to be
    polynomial in the size of the input and the \e{size} of the
    previously produced results (not just their \e{number} $N$ as we
    define here). However, the definitions are equivalent when the
    size of each answer is polynomial in that of the input, as in our
    case.}
\end{itemize}
Observe that a solver that enumerates with polynomial delay also
enumerates with incremental polynomial time, which, in turn, implies
polynomial total time.

\eat{
Later on, we will use the following insight. Let $\P$ be an
enumeration problem, and let $A$ be a solver for $\P$.  For input $x$
and answer $y\in\P(x)$, we denote by $\time_{A,x}(y)$ the time in
which $y$ is printed.

\begin{proposition}
  Let $\P$ be an enumeration problem, and $A$ a solver for
  $\P$. Suppose that 
}

\subsection{Enumerating the Minimal Triangulations}\label{sec:poly}
A common approach to establish enumeration with polynomial delay is
via the technique known as the \e{branch-and-bound} (or the
\e{flashlight}) method~\cite{DBLP:conf/esa/BorosEG04}. In this
approach, we find a condition $\psi$ over the answers, and then
recursively enumerate all of the answers that satisfy $\psi$ and all
of the answers that violate $\psi$ (i.e., satisfy
$\psi'=\neg\psi$). Hence, in each recursive call, we need to enumerate
all the answers that satisfy a conjunction
$\psi_1\land\dots\land\psi_m$ of such conditions. For this approach to
guarantee polynomial delay, the depth of the recursion should be
bounded by a polynomial in the size of the input. Importantly, in each
recursive call, we should be able to test whether there is \e{at least
  one} answer that satisfies $\psi_1\land\dots\land\psi_m$. Then, in
the leaves, we should be able to produce the single answer that
satisfies the given constraints.

In this manuscript, we devise an algorithm for enumerating the minimal
triangulations: given $g$, enumerate $\mintri(g)$. A branch-and-bound
attempt to solve this problem would be, say, to apply the conditions
of inclusion and exclusion of fill edges. This approach amounts to
testing whether there is a minimal triangulation that contains a given
set of edges and excludes another given set of edges. Unfortunately,
it follows from known hardness results of Golumbick, Kaplan and
Shamir~\cite{GOLUMBIC1995449} that this problem is intractable.

\begin{proposition}
  The following decision problem is NP-complete: given $g$ and two
  sets $I$ and $X$ of node pairs, is there a minimal triangulation $h$
  of $g$ such that $I\subseteq\edges(h)$ and $\edges(h)\cap
  X=\emptyset$?
\end{proposition}
\begin{proof}
  Membership in NP is straightforward. To show hardness, we use a
  reduction from the \e{chordal sandwich problem}. For a graph
  property $\Pi$, the sandwich problem for $\Pi$ is that of
  determining, given graphs $g$ and $g''$ with $\nodes(g)=\nodes(g'')$
  and $\edges(g)\subseteq \edges(g'')$, where there exists a graph
  $g'$ such that $\nodes(g')=\nodes(g)$, $\edges(g)\subseteq
  \edges(g')\subseteq \edges(g'')$, and $g'$ satisfies
  $\Pi$. Golumbick et al.~\cite{GOLUMBIC1995449} proved the
  NP-hardness of this problem for various graph properties $\Pi$,
  including chordality. Now, given $g$ and $g''$, let $X$ be the set
  of all the node pairs that are \e{not} edges of $g''$. The existence
  of $g'$ in the chordal sandwich problem is equivalent to the
  existence of a (minimal) triangulation of $g$ that excludes $X$.
\end{proof}

Hence, we adopt a different approach to enumerating the minimal
triangulations, as we describe in the following sections.

\section{Enumerating Maximal Independent Sets on Succinct
  Graphs}\label{sec:SGR}

The main result of this paper is an algorithm for enumerating the
minimal triangulations of a graph $g$. As we explain in the next
section, this problem amounts to enumerating the maximal independent
sets of a graph $h$. It is known that the maximal independent sets of
a graph can be enumerated with polynomial
delay~\cite{DBLP:journals/ipl/JohnsonP88}. However, we cannot
instantiate $h$, since the number of nodes of $h$ can be exponential
in the size of $g$. Hence, known algorithms for enumerating maximal
independent sets cannot be applied to solve our problem. Nevertheless,
$h$ possesses some tractability properties that, in fact, allow us to
efficiently enumerate the maximal independent sets of $h$. In this
section we identify these properties within an abstract framework of
\e{succinct graph representations}, where a graph may be exponentially
larger than its representation, and we have access to the nodes and
edges through efficient algorithms. Mainly, we devise an algorithm for
enumerating the maximal independent sets for such graphs.

\subsection{Succinct Graph Representations}

We begin with the formal definition of a succinct graph
representation.
\begin{definition}[SGR]
  A \e{Succinct Graph Representation} (\e{SGR}) is a triple
  $(\G,\Av,\Ae)$, where:
\begin{itemize}
\item $\G$ is a function that maps every input $x$, referred to as an
  \e{instance}, to a graph $\G(x)$;
\item $\Av$ is an enumeration algorithm that, given an instance $x$,
  enumerates the nodes of $\G(x)$;
\item $\Ae$ is a decision algorithm that, given an instance $x$ and
  nodes $v$ and $u$ of $\G(x)$, determines whether $v$ and $u$ are
  connected by an edge in $\G(x)$.%
\end{itemize}
\end{definition}

An SGR $(\G,\Av,\Ae)$ is said to be \e{tractably accessible} if both
the following hold.
\begin{enumerate}
\item $\Av$ enumerates with polynomial delay.
\item $\Ae$ terminates in polynomial time.
\end{enumerate}
Here, both polynomials are with respect to $|x|$ (the length of $x$).
Observe that in a tractably accessible SGR, the (representation) size
of every node $v$ of $\G(x)$ is polynomial in that of $x$ (since
writing $v$ is within the polynomial delay).

 For efficient enumeration of $\maxind(\G(x))$, we need some more
 tractability conditions.
\begin{definition}[Tractable Expansion]\label{def:tractable-expansion}
  \hskip0.5em A tractably accessible SGR $(\G,\Av,\Ae)$ is said to
  have a \e{tractable expansion} if both of the following conditions
  hold.
\begin{enumerate}
\item There is a polynomial $p$ such that $|I|\leq p(|x|)$ for all
  instances $x$ and independent sets $I$ of $\G(x)$.
\item There is a polynomial-time algorithm that, given $x$ and an
  independent set $I$ of $\G(x)$, either determines that $I$ is
  maximal or returns a node $v\notin I$ such that
  $I\cup\set{v}$ is
  independent.
\end{enumerate}
\end{definition}

Following is an example
of an SGR that is central to this paper.

\subsubsection{The Separator Graph as an SGR}\label{sec:msgraph}

The \e{separator graph} of a graph $g$ is the graph that has the set
$\minsep(g)$ of minimal separators as its node set, and an edge
between every two minimal separators that are crossing (i.e.,
$S,T\in\minsep(g)$ such that $S\crosses T$).  Throughout this paper
we denote by $\sepsys$ the SGR $(\G\ms,\Av\ms,\Ae\ms)$, where:
\begin{itemize}
	\item $\G\ms$ is a function that maps the representation of a graph
	$g$ to its separator graph $\G\ms(g)$.
      \item $\Av\ms$ is an enumeration algorithm that, given a graph $g$,
        enumerates its set $\minsep(g)$ of minimal separators. We can
        use here a variation of the algorithm of Berry et
        al.~\cite{conf/wg/BerryBC99} that enumerates $\minsep(g)$ with
        polynomial delay, as we describe later in
        Section~\ref{sec:min-sep-enum}.
	\item $\Ae\ms$ is an algorithm that, given a graph $g$ and two minimal
	separators $S$ and $T$, determines whether $S\crosses T$ efficiently
	(e.g., by removing $S$ and testing whether $T$ is split along
	multiple connected components).
\end{itemize} 
In particular, $\sepsys$ is a tractably accessible SGR.

\subsection{Enumerating Maximal Independent Sets in
  SGRs}\label{sec:enumSGR}

Our main result for this section is the following.
\begin{theorem}\label{thm:sgr-inc-ptime}
  Let $(\G,\Av,\Ae)$ be a tractably accessible SGR with a tractable expansion. There is an algorithm that, given an
  instance $x$, enumerates the set $\maxind(\G(x))$ in incremental polynomial
  time.
\end{theorem}

The proof is via the algorithm \algname{EnumMIS} that is depicted in
Figure~\ref{alg:EnumMaxIndependent}.  \addedb{ This algorithm is an
  adaptation of the algorithm for computing full disjunctions in
  databases~\cite{DBLP:conf/vldb/CohenFKKS06} that generalizes the
  problem of enumerating maximal cliques (or maximal independent
  sets).  In turn, that algorithm was based on an improvement of the
  algorithm of Lawler et al.~\cite{DBLP:journals/siamcomp/LawlerLK80}
  for generating the maximal independent sets in polynomial total
  time, and all rely on the general idea of reducing the problem to
  the \e{input-restricted problem} that was later introduced by Cohen
  et al.~\cite{DBLP:journals/jcss/CohenKS08} for enumerating maximal
  node sets that satisfy a hereditary property.}

\addedb{
  The underlying idea of that algorithm is to construct a graph over
  the space of the solutions (maximal independent sets), and traverse
  the graph in a depth-first-search manner. In the case of maximal independent sets, there is an edge
  from $J$ to $K$ if $K$ is obtained from $J$ by adding a new node
  $v$, removing the neighbours of $v$, and greedily extending to a
  maximal independent set.  }

In this section, we describe the algorithm and prove its correctness
and efficiency. In the remainder of this section, we fix a tractably
accessible SGR $(\G,\Av,\Ae)$ with tractable expansion, and an input
instance $x$. Our goal is to enumerate $\maxind(\G(x))$.

\begin{algseries}{H}{Enumerating maximal independent sets of an input
    $x$ for an SGR $(\G,\Av,\Ae)$\label{alg:EnumMaxIndependent}}
\begin{insidealg}{EnumMIS}{$x$}
\STATE $J\asn\algname{Extend}(x,\emptyset)$
\STATE \textbf{print} $J$\label{algline:print1}
\STATE $\Q\asn\set{J}$
\STATE $\MO\asn\emptyset$
\STATE $\V\asn\emptyset$
\STATE $\iter\asn\Av(x)$
\WHILE{$\Q\neq\emptyset$} \label{algline:mainLoopStart}
\STATE $J\asn\Q.\algname{pop}()$ \label{algline:pop}
\STATE $\MO.\algname{push}(J)$ \label{algline:push}
\FORALL{$v\in\V$} \label{algline:nodeLoop1}
\STATE $J_v\asn\set{v}\cup\set{u\in J\mid \neg\Ae(x,v,u)}$ \label{algline:innerStart}
\STATE $K\asn\algname{Extend}(x,J_v)$ \label{algline:genMIS1}
\IF{$K\notin\Q\cup\MO$}  \label{algline:checkExistStart}
\STATE \textbf{print} $K$\label{algline:print2}
\STATE $\Q\asn\Q\cup\set{K}$ \label{algline:newTriangulation1}
\ENDIF   \label{algline:checkExistFinish}
\ENDFOR
\WHILE{$\Q=\emptyset$ and $\iter.\algname{hasNext}()$} \label{algline:hasNextLoop}
\STATE $v \asn {\iter.\algname{next}()}$ \label{algline:generateNode}
\STATE $\V\asn\V\cup\set{v}$  \label{algline:addtoV}
\FORALL{$J'\in \MO$} \label{algline:printedLoop}
\STATE $J'_v\asn\set{v}\cup\set{u\in J'\mid \neg\Ae(x,v,u)}$ \label{algline:exp2}
\STATE $K\asn\algname{Extend}(x,J'_v)$ \label{algline:genMIS}
\IF{$K\notin\Q\cup\MO$}  \label{algline:checkExistBeg}
\STATE \textbf{print} $K$\label{algline:print3}
\STATE $\Q\asn\Q\cup\set{K}$  \label{algline:newTriangulation2}
\ENDIF \label{algline:checkExistEnd}
\ENDFOR
\ENDWHILE
\ENDWHILE \label{algline:mainLoopEnd}
\end{insidealg}
\label{alg:EnumMIS}
\end{algseries}
 
\subsubsection{Algorithm Description}
As explained earlier, the algorithm extends every maximal independent
set $J$ that it generates in the direction of every node $v$ that it
generates. By \e{extending $J$ in the direction of $v$} we mean
producing an arbitrary maximal independent set $K$ that contains $v$
and all nodes in $J$ that are non-neighbors of $v$.
As long as there are unprocessed sets, they are extended in the
direction of all previously generated nodes. When no unprocessed sets
are left, additional nodes are generated, and the previously processed
sets are extended in the direction of the new nodes.  \addedb{Put
  differently, our algorithm adapts the traversal approach by
  restricting the steps to the solutions that are obtained by
  extending in the direction of the nodes $v$ that have been produced
  until that point of time; when a new node $v$ is generated, we
  revisit the past solutions and take the steps implied by $v$.
}

The algorithm maintains two collections, $\Q$ and $\MO$, for storing
answers (which are maximal independent sets of the graph $\G(x)$). The algorithm
inserts answers into $\Q$, and repeatedly \e{removes} (or \e{pops}) an
answer from $\Q$ and \e{processes} that answer (while possibly
inserting new answers into $\Q$), until $\Q$ is empty.  The set $\MO$
stores the answers that have already been removed from $\Q$ and
processed.  Importantly, both collections feature membership testing,
element removal and element insertion with a number of comparisons
logarithmic in their cardinality (i.e., the number of answers they
hold at the time of the operation). In addition, the algorithm
maintains a collection $\V$ of the nodes of $\G(x)$ generated thus
far.  The collection $\Q$ is initialized with a single result (which
is printed after being generated), which is an arbitrary maximal independent
set. This result is obtained through the procedure
$\algname{Extend}(x,I)$ that extends a given independent set $I$ into
a maximal one. Note that this procedure can be implemented in
polynomial time, since $(\G,\Av,\Ae)$ has a tractable expansion.  The
sets $\MO$ and $\V$ are initialized empty.

The algorithm accesses the nodes of $\G(x)$ through an iterator object
that is obtained by executing $\Av(x)$, and features two
polynomial-time operations:
\begin{itemize}
\item Boolean $\algname{hasNext}()$ determines whether there are
  additional nodes of $\G(x)$ to enumerate.
\item $\algname{next()}$ returns the next node in the iteration.
\end{itemize}

The algorithm applies the iteration of
line~\ref{algline:mainLoopStart} until $\Q$ becomes empty, and then terminates. In every iteration, the algorithm pops an element
from $\Q$, stores it in $\MO$
(lines~\ref{algline:pop}--\ref{algline:push}), and then processes it. The algorithm
iterates through the nodes in $\V$, and for each node $v$ it applies
(in lines~\ref{algline:innerStart}--\ref{algline:checkExistFinish})
what we call \e{extension of $J$ in the direction of $v$}: 
\begin{enumerate}
\item Generate the set $J_v$ that consists of $v$ and all the nodes in
  $J$ that are non-neighbors of $v$, using the algorithm $\Ae$ for
  testing adjacency;
\item Extend $J_v$ into an arbitrary maximal independent set $K$,
  again using $\algname{Extend}(x,J_v)$;
\item If $K$ is in neither $\Q$ nor $\MO$ (meaning it was not printed before), then print $K$ and add it
  to $\Q$.
\end{enumerate}
Observe that $J_v$ is an independent set, and therefore, it is possible to
invoke $\algname{Extend}(x,J_v)$ with $J_v$.

Up to this point, the algorithm is very similar to the algorithm of
Cohen et al.~\cite{DBLP:conf/vldb/CohenFKKS06} for computing full
disjunctions, except that $\V$ does not hold all nodes but only the
nodes generated so far. The twist (and the source of extra challenge
in proving correctness and efficiency) is in
lines~\ref{algline:hasNextLoop}--\ref{algline:mainLoopEnd}, where we
generate additional nodes and compensate for them being missing in the
previous iterations. In these lines, the algorithm tests whether it is
the case that $\Q$ is empty and the node iterator has additional nodes
to process (line~\ref{algline:hasNextLoop}). While this is the case,
the algorithm repeats the following procedure
(lines~\ref{algline:generateNode}--\ref{algline:checkExistEnd}):
generate the next node using the iterator of $\Av(x)$, add it to $\V$,
and extend \e{every} previously processed result (i.e., the results in
$\MO$) in the direction of the newly generated node $v$ (as previously
described).

\subsubsection{Correctness and Efficiency}\label{sec:sgr-correct}

The following lemma states the correctness of the algorithm: the
algorithm enumerates \e{every} element in $\maxind(\G(x))$, \e{only}
elements in $\maxind(\G(x))$, and every element is printed exactly
once.

\def\lemmasgrenumcorrect{
\algname{EnumMIS}$(x)$ enumerates $\maxind(\G(x))$.
}

\begin{lemma}\label{lemma:sgr-enum-correct} 
\lemmasgrenumcorrect
\end{lemma}

{\begin{qproof}
	The algorithm prints only elements that are created by invoking the
	procedure $\algname{Extend}$. Therefore, the algorithm prints only
	elements in $\maxind(\G(x))$. The tests of
	lines~\ref{algline:checkExistStart} and \ref{algline:checkExistBeg}
	ensure that whenever an element is printed, this element has not
	been seen before. Hence, no element is printed more than once. It is
	left to prove that every maximal independent set of $\G(x)$ is
	printed by the algorithm.
	
	Observe the following. When the algorithm terminates we have
	$\Q=\emptyset$. Therefore, in the previous iteration the loop of
	line~\ref{algline:hasNextLoop} could only have terminated due to
	$\iter.\algname{hasNext}()$ returning false. Therefore, upon
	termination $\V = \nodes(\G(g))$.
	
	Suppose, by way of contradiction, that there is some maximal
	independent set $H$ that is not printed by the algorithm. Let $J$ be
	a maximal independent set of $\G(x)$, among all the printed ones,
	that contains a maximal number of elements from $H$. The set $J$ must exist, since the algorithm prints at least one maximal independent set. Let $H_m$ be
	the intersection $H\cap J$. Since $H\neq
	H_m$ (or else $H$ is not maximal), there is at least one node in $H
	\setminus J$; let $v$ be such a node.
	
	At this point we have established that before the algorithm
	terminated, \e{(a)} the node $v$ has been generated, \e{and (b)} $J$
	has been printed. We now branch into two cases, as follows.
	\begin{enumerate}
        \item The set $J$ was inserted into $\MO$ before the node $v$
          was generated. Immediately after $v$ is generated (in line
          \ref{algline:generateNode}), the set $J_v=\set{v} \cup \{u
          \in J\mid \neg\Ae(x,v,u)\}$ will be constructed (in
          line~\ref{algline:exp2}) and expanded to a maximal
          independent set $K$ that contains $J_v$.
		\item The node $v$ was generated before $J$ was inserted into $\MO$.
		At the iteration when $J$ is inserted into $\MO$, we have
		$v \in \V$, and so the set $J_v=\set{v} \cup \{u \in J\mid
		\neg\Ae(x,v,u)\}$ will be constructed (in line~\ref{algline:innerStart}) and expanded to a maximal
		independent set $K$ that contains $J_v$.
	\end{enumerate}
	
	So, we have established that before the algorithm terminates, the set
	$J_v$ is generated and expanded to a maximal independent set $K$
	that contains $J_v$. Furthermore, $H_m \cup \set{v} \subseteq J_v$
	(since $H_m \subseteq J$, and does not contain any neighbor of $v$), and therefore $H_m
	\cup \set{v}\subseteq K$.  According to the algorithm, one of the
	following options must hold: (1) $K$ is inserted into $\Q$, (2) $K$ is
	already in $\Q$ (3) $K$ was in $\Q$ in the past and is now in
	$\MO$. Since the algorithm prints every maximal independent set that
	is inserted into $\Q$, we get a contradiction to the maximality of
	$H_m$.
\end{qproof}}

We now prove that the algorithm $\algname{EnumMIS}$ enumerates with
incremental polynomial time. We do so in two steps. We first define an
algorithm that is similar to $\algname{EnumMIS}$, but with a small
twist that makes it easier to prove incremental polynomial time.
Then, we prove a general result that will imply that, if the new
algorithm enumerates in incremental polynomial time then so does
$\algname{EnumMIS}$.

The new algorithm is similar to $\algname{EnumMIS}$, except that each
of the print commands (lines~\ref{algline:print1},
\ref{algline:print2} and \ref{algline:print3}) is replaced with an
operation that takes the time of the printing, but is actually void
(e.g., printing to \texttt{/dev/null} in Unix). Instead, each maximal
independent set is printed immediately after being removed from $\Q$
(line~\ref{algline:pop}).  Hence, answers are \e{held} until removed
from $\Q$. \addedb{We refer to the resulting algorithm as
$\algname{EnumMISHold}$.} 
\addedn{For theory purposes, it would have been enough to discuss only $\algname{EnumMISHold}$, which is easier to analyze. However, since delaying the results is not required to obtain the theoretical guarantees, we also discuss $\algname{EnumMIS}$ where we print the results as soon as we have them.}
Next, we prove that $\algname{EnumMISHold}$
enumerates in incremental polynomial time. Observe that to bound the
delay $\algname{EnumMISHold}$, we only need to bound the time between
two executions of line~\ref{algline:pop} of $\algname{EnumMIS}$.

\begin{lemma}\label{lemma:sgr-enum-hold}
  $\algname{EnumMISHold}(x)$ enumerates with incremental polynomial
  time.
\end{lemma}
\begin{qproof}
  
We begin by showing that the size of the node set $\V$ is polynomial in the size of the printed result set $\MO$.
Whenever a new node $v$ is inserted into $\V$ (line~\ref{algline:addtoV}), the set $\Q$ is empty.
The following calls to \algname{Extend} (line~\ref{algline:genMIS}) will generate maximal independent sets containing $v$.
Each of these maximal independent sets is either already in $\MO$, or it is inserted into $\Q$ (line~\ref{algline:newTriangulation2}).
Therefore, at the end of the iteration of the main loop in which $v$ was inserted into $\V$, all maximal independent sets in $\Q$ contain $v$. In the next iteration of the main loop, if such an iteration exists, one of these newly generated independent sets will be printed and inserted into $\MO$ (line~\ref{algline:push}).
That is, at the beginning of every iteration of the algorithm (specifically, line~\ref{algline:nodeLoop1}), every node $v \in \V$ belongs to some maximal independent set that has already been printed (and thus part of $\MO$). Since we assume tractable expansion, each independent set in $\MO$ contains at most $p(|x|)$ nodes, and we can conclude that $|\V| \leq p(|x|)\cdot|\MO|$.

We now bound the time between two executions of
line~\ref{algline:pop} of
$\algname{EnumMIS}$. Line~\ref{algline:push} takes polynomial time
in $|x|$ (since there are at most exponentially many independent
sets, $(\G,\Av,\Ae)$ has a {tractable expansion}, and operations
on $\MO$ require a logarithmic number of comparisons in the cardinality).  The number of
iterations of line~\ref{algline:nodeLoop1} is at most the size of
$\V$, which is polynomial in the number of answers printed so far (due
to the above observation). Each operation in that iteration takes
time polynomial in $|x|$.

The loop of line~\ref{algline:hasNextLoop} repeats (at most) until a
node that belongs to none of the printed answers is generated.  Hence,
the observation that this number is polynomial in the size of the
output, along with the tractable expansion, again implies that we
iterate a number of times that is polynomial in the number of answers
printed so far. The loop of line~\ref{algline:printedLoop} repeats at
most as many times as the number of answers in $\MO$, and all of these
have been printed before. Besides the loops, each of
lines~\ref{algline:generateNode}--\ref{algline:mainLoopEnd} takes
polynomial time in $|x|$. %
\end{qproof}

Lemma~\ref{lemma:sgr-enum-hold} shows that $\algname{EnumMISHold}$
enumerates with incremental polynomial time. Next, we show the same
for $\algname{EnumMIS}$. The key point is that every answer is printed
in $\algname{EnumMIS}$ \e{no later} than it is printed in
$\algname{EnumMISHold}$.  \addedb{Note that this holds even though the two
algorithms do not necessarily enumerate in the same order (as we
make no assumptions about the order of removal in $\Q$), since we
assume that $\algname{EnumMISHold}$ spends on void the printing time
of $\algname{EnumMIS}$.}
We will prove that this suffices to conclude
that if $\algname{EnumMISHold}$ enumerates in
incremental polynomial time, then so does $\algname{EnumMIS}$.  We
prove here a general result. Let $\P$ be an enumeration problem, and
let $A$ be a solver for $\P$.  For input $x$ and answer $y\in\P(x)$,
we denote by $\time_{A,x}(y)$ the time in which $y$ is printed. We
prove the following theorem.

\def\thmincimpliesinc{
 Let $\P$ be an enumeration problem, and let $A$ and $B$ be two
  solvers for $\P$. Suppose that for all instances $x$ and for all answers $y\in\P(x)$
  we have $\time_{A,x}(y)\leq\time_{B,x}(y)$. If $B$ enumerates in
  incremental polynomial time, then so does $A$.
}

\begin{theorem}\label{thm:inc-implies-inc}
\thmincimpliesinc
\end{theorem}

Theorem~\ref{thm:inc-implies-inc} is not a vacuous statement, since
the order of results may differ between $A$ and $B$. Furthermore, the
corollary no longer holds when substituting ``incremental polynomial
time'' with ``polynomial delay.''
\addedn{
For example, imagine two algorithms that print all subsets of an input set. The first prints a new answer after every two time ticks, while the second prints them after every single time tick, except for the last answer which is printed at the same time in both algorithms.
The first algorithm meets the guarantee of polynomial delay, and even though the second algorithm prints every answer no later than the first, the second algorithm does not enumerate in polynomial delay as its delay before the last answer is exponential.
}

Let $\P$ be an enumeration problem, let $A$ be a solver for $\P$, and
let $x$ be input for $A$. If $\tau$ is a time tick during the
execution of $A(x)$, then we denote by $\answers_{A,x}(\tau)$ the
answers $y\in\P(x)$ that have been printed before time $\tau$ is
reached.  We have the following lemma.

\def\lemmaequivinc{
	Let $\P$ be an enumeration problem, and $A$ a solver for $\P$. The
	following are equivalent.
	\begin{enumerate}
		\item $A$ enumerates in incremental polynomial time.
		\item There is a polynomial $p$ such that for all input $x$
                  and time tick $\tau$ it holds that
                  \[p(|x|+|\answers_{A,x}(\tau)|)>\tau\,.\]
	\end{enumerate}
}

\begin{lemma}\label{lemma:eqiv-inc}
	\lemmaequivinc
\end{lemma}
\begin{qproof}
	Denote the time of the $N$th result by $t_N$.
	\partitle{$\mathbf{ 1 \Rightarrow 2}$} If $A$ enumerates in incremental
        polynomial time, there exists a polynomial $p_1$ such that
        $t_{N+1}-t_N \leq p_1(|x|+N)$. Without loss of generality, we
        assume that $p_1$ is monotone (as every polynomial is upper
        bounded by some monotone polynomial, and we can replace $p_1$
        with such polynomial). We get the following on the printing
        time of the $N$th result.
	\begin{equation*}
	t_N = \sum_{i=1}^{N}{t_i-t_{i-1}} 
	\leq \sum_{i=1}^{N}{p_1(|x|+i-1)} 
	\leq N\cdot p_1(|x|+N-1)
	\end{equation*}
	In this case we get that for any time $\tau$ there exists a polynomial $p_2$ such that the following holds.
	\begin{equation*}
	\begin{split}
	\tau &< t_{|\answers_{A,x}(\tau)|+1} \leq (|\answers_{A,x}(\tau)|+1)\cdot p_1(|x|+|\answers_{A,x}(\tau)|)  \\
	&\leq p_2(|x|+|\answers_{A,x}(\tau)|)
	\end{split}
	\end{equation*}
	\partitle{$\mathbf{ 2 \Rightarrow 1}$} Assume now that $p_3(|x|+|\answers_{A,x}(\tau)|)>\tau$ for any time $\tau$.
	Consider the delay after the $N$th answer.
	\begin{equation*}
	t_{N+1}-t_N \leq t_{N+1} < p_3(|x|+N+1)
	\end{equation*}
	This shows that there exists a polynomial $p_4$ such that 
	$t_{N+1}-t_N < p_4(|x|+N)$,
	meaning that $A$ enumerates in incremental polynomial time.
\end{qproof}

We can now prove Theorem~\ref{thm:inc-implies-inc}.

\begin{proof}
  Using the characterization of Lemma~\ref{lemma:eqiv-inc}, let $p$ be
  a polynomial such that for all $x$ and $\tau$ we have
  $p(|x|+|\answers_{B,x}(\tau)|)>\tau$.  The condition
  of the theorem implies that at every time tick $\tau$, the set of
  answers printed by $B$ is a subset of the set of answers printed by
  $A$, and therefore, $|\answers_{A,x}(\tau)|\geq
  |\answers_{B,x}(\tau)|$. Again since we can assume monotonicity, we conclude that
  $p(|x|+|\answers_{A,x}(\tau)|)>\tau$ as well. We use
  Lemma~\ref{lemma:eqiv-inc} to conclude that $A$ enumerates in
  incremental polynomial time.
\end{proof}

Using the algorithms $\algname{EnumMIS}$ and $\algname{EnumMISHold}$
as $A$ and $B$ in Theorem~\ref{thm:inc-implies-inc}, respectively, the
combination with Lemma~\ref{lemma:sgr-enum-hold} implies that
$\algname{EnumMIS}$ enumerates in incremental polynomial time, as
claimed.

\subsection{Tightness of the Algorithm}\label{sec:sgrpdelay}

In the following, we show that the time bounds that
$\algname{EnumMIS}$ achieves are tight since it is not possible to
solve the same problem with polynomial delay under the $\SETH$
assumption.

We recall that $k$-SAT is the satisfiability problem over $n$
variables, where every clause contains at most $k$ literals. $\SETH$
states that there is no algorithm for solving $k$-SAT in
$O^{*}(2^{(1-\varepsilon)n})$ time for a fixed $\varepsilon$ and all
$k$\addedn{, where the $O^{*}$-notation omits polynomial factors}.
\begin{definition}[The Strong Exponential Time Hypothesis]
For every $\varepsilon > 0$ there exists a $k$ such that $k$-SAT requires time larger than $2^{(1-\varepsilon)n}$ where $n$ is the number of variables.
\end{definition}
\addedn{
Let $(\calG,A_V,A_E)$ be the tractable expansion of a tractably accessible SGR. We denote by
$\mathsf{SMIS(\calG,A_V,A_E)}$ the following enumeration problem: Given an instance $x$,
enumerate all $\maxind(\G(x))$.
}
\begin{proposition}\label{prop:lowerbound}
There exists some tractably accessible SGR with a tractable expansion $(\calG,A_V,A_E)$,
such that \addedn{$\mathsf{SMIS(\calG,A_V,A_E)}$} cannot be enumerated with
a polynomial delay, assuming the $\SETH$.
\end{proposition}
\begin{proof}
Let $k\geq 3$, and
let $\phi$ be an instance of $\kSAT$ with $\var(\phi)=\{x_1,\ldots,x_n\}$ (for
readability, we assume that $n\geq 2$ is even). We will
show that a polynomial delay algorithm for enumerating $\maxind(\G(x))$ will decide satisfyability
of $\phi$ within time $2^{\sfrac{n}{2}}\cdot\poly(|\phi|)$. This is true for any choice
of $k$, which is not
possible assuming the $\SETH$.

We first describe the SGR $(\calG,A_V,A_E)$. For any string $x$ that is not a $\kSAT$ formula,
$\calG(x)=\emptyset$. Otherwise, given a $\kSAT$ instance $\phi$, we define $\calG(x)$ as follows:
The set of vertices represents all possible truth assignments on $\frac{n}{2}$ variables twice,
with two additional nodes $\bot_A$ and $\bot_B$. Intuitively, 
$V_A$ corresponds to all possible truth assignments on the variables $x_1,\ldots,x_{\frac{n}{2}}$,
and $V_B$ corresponds to all possible truth assignments on the remaining variables
$x_{\frac{n}{2}+1},\ldots,x_n$.
That is,
\begin{align*}
V_A &= \{A\}\times\{0,1\}^{\frac{n}{2}}\\
V_B &= \{B\}\times\{0,1\}^{\frac{n}{2}}\\
V(\calG(\phi)) &= V_A \cup V_B \cup \{\bot_A\}\cup\{\bot_B\}.
\end{align*}

To define the set of edges, we first start with edges between the set $V_A$ and $V_B$.
There is an edge $(u,v)$ for $u\in V_A$, $v\in V_B$ if and only if $u$ and $v$ together encode a truth
assignment
that does not satisfy $\phi$. Moreover, we also add all edges between nodes in $V_A$, between nodes
in $V_B$ and certain connections to the nodes $\bot_A$ and $\bot_B$ as follows:
\begin{align*}
E_{unsat} &= \{\{u,v\}|\exists a_1,\ldots,a_n \in\{0,1\}\text{ s.t. }
u=(A,a_1,\ldots,a_{\frac{n}{2}})\in V_A, \\
&\phantom{= \{\{u,v\}|}v=(B, a_{\frac{n}{2}+1},\ldots,a_n)\in V_B \text{ and } \phi(a_1,\ldots,a_n)=\false\}.\\
E(\calG(\phi)) &=  E_{unsat} \cup \{\bot_A,\bot_B\} \cup \{\{u,v\}\mid u,v \in V_A\} \cup \{\{u,v\}\mid u,v \in V_B\}\\
 &\phantom{=} \cup  \{\{u, \bot_A\}\mid u\in V_A\} \cup \{\{u, \bot_B\}\mid u\in V_B \}
\end{align*}
We first note that this SGR is tractably accessible. Indeed, the set of nodes can be enumerated with a
polynomial (even constant) delay, and for any $u,v\in \calG(\phi)$, we can check whether
$\{u,v\}\in E(\calG(\phi))$ in polynomial time, since evaluation of any $\kSAT$ formula can be done within polynomial
(or even linear) time. To show that this SGR also has a tractable expansion, we note that the set of
maximal independent sets of $\calG(\phi)$ is given as the union of the sets $I_A, I_B$ and $I_{sat}$
with
\begin{align*}
I_A &= \{\{u, \bot_B\} \mid u\in V_A\}, \quad I_B = \{\{u,\bot_A\}\mid u\in V_B\}\text{ and }\\
I_{sat} &=  \{\{u,v\}|\exists a_1,\ldots,a_n \in\{0,1\}\text{ s.t. }
u=(A,a_1,\ldots,a_{\frac{n}{2}})\in V_A, \\
&\phantom{= \{\{u,v\}|}v=(B, a_{\frac{n}{2}+1},\ldots,a_n)\in V_B \text{ and } \phi(a_1,\ldots,a_n)=\true\}.\\
\end{align*}
Every maximal independent set of $\calG(\phi)$ is of size $2$, satisfying the first condition of
a tractable expansion. For the second condition, we note that every subset $I$ of $V(\calG)$ of
size one can be extended trivially to a maximal independent set
(by adding either $\bot_A$, $\bot_B$, or in case that $I\subset\{\bot_A,\bot_B\}$ some arbitrary element
from $V_A$ or $V_B$), and for any subset of size two, we can check whether $I$ is (maximally) independent
within polynomial time. 

Note that $\phi$ is satisfiable if and only if $\maxind(\G(x))$ contains more than the
sets $I_A$ and $I_B$.
Assume that we can enumerate $\maxind(\G(x))$ with a polynomial delay. We can output $2\cdot 2^{\frac{n}{2}}$
many solutions within time $2^{\frac{n}{2}}\cdot\poly(|\phi|)$, meaning that
we can decide whether there are more than $2\cdot 2^{\frac{n}{2}}$ many maximal independent sets
of $\calG(\phi)$ within in the same time bound. Since $\phi$ is satisfyable iff $\calG(\phi)$
has at least $2\cdot 2^{\frac{n}{2}} + 1$ maximal independent sets, we are done. 
\end{proof}

\subsection{Note on Space Usage}

\addedb{ We conclude this section with a discussion on the space
  usage. Note that our algorithm may reach an exponential space as it
  relies on remembering all past answers to avoid the production of
  duplications. This cost is already incurred in the enumerators of
  maximal independent sets that form the basis of our
  algorithm~\cite{DBLP:conf/vldb/CohenFKKS06,DBLP:journals/jcss/CohenKS08,DBLP:journals/siamcomp/LawlerLK80}.
  However, several algorithms for enumerating maximal independent sets
  (and more generally maximal sets w.r.t.~different properties)
  guarantee both polynomial delay and polynomial space, including the
  \e{reverse search}~\cite{DBLP:journals/dam/AvisF96}, the algorithm
  of Conte et al.~\cite{DBLP:conf/spire/ConteGMUV17}, and the
  \e{proximity search}~\cite{DBLP:conf/stoc/ConteU19}. However, it
  is not clear to us how these algorithms can be adapted to
  enumerating the maximal independent sets of an SGR in a manner that
  limits the space, given that the set of nodes is not known upfront
  (and in light of Proposition~\ref{prop:lowerbound}).  Moreover, note
  that the exponential space of our algorithm is also required for
  storing the (possibly exponential number of) past generated nodes of
  the SGR.}

\addedb{ A natural question then remains open: can
  Theorem~\ref{thm:sgr-inc-ptime} be improved to require only
  polynomial space (at least when ignoring the space used by invoking
  the SGR functions)? Particularly, we leave open the question of
  whether and how the aforementioned polynomial-space algorithms can
  be adapted to enumerating the maximal independent sets of an SGR,
  and whether we can avoid storing all produced nodes. It appears that
  further assumptions on the SGR are required to this aim, and
  establishing these assumptions is left as a future direction.}

\section{Enumerating Minimal Triangulations}\label{sec:enumMT}

In Section~\ref{sec:msgraph} we introduced $\sepsys$ and claimed that
it is an SGR.  In this section, we will use known results to reduce
the problem of enumerating the minimal triangulations of a graph to
the problem of enumerating the maximal independent sets for $\sepsys$.
We will describe how to enumerate the nodes of $\sepsys$ with
polynomial delay, concluding that it is in fact an SGR.  We will
further show that $\sepsys$ has a tractable expansion
(Definition~\ref{def:tractable-expansion}), and therefore
Theorem~\ref{thm:sgr-inc-ptime} can be applied to conclude that the
minimal triangulations can be enumerated in incremental polynomial
time.

\subsection{Reduction}
We use the following notation. Let $g$ be a graph.  We denote by
$\clqminsep(g)$ the set of minimal separators $S$ of $g$, such that
$S$ is a clique of $g$. Let $\sepset$ be a subset of $\minsep(g)$. We
denote by $g_{[\sepset]}$ the graph that results from saturating the
minimal separators in $\sepset$.

Parra and Scheffler~\cite{DBLP:journals/dam/ParraS97} have shown the
following connection between minimal triangulations and maximal sets
of \e{pairwise-parallel minimal separators} (that is, every two
minimal separators in the set are non-crossing).

\begin{citedtheorem}{Parra and
    Scheffler~\cite{DBLP:journals/dam/ParraS97}}\label{thm:ParraS97}
  Let $g$ be a graph.
	\begin{enumerate}
        \item \label{item:ParraS1st} If $\sepset$ is a maximal set of pairwise-parallel
          minimal separators of $g$, then $g_{[\sepset]}$ is a minimal
          triangulation of $g$, and $\minsep(g_{[\sepset]})=\sepset$.
        \item \label{item:ParraS2nd} If $h$ is a minimal triangulation of $g$, then the set
          $\sepset=\minsep(h)$ is a maximal set of pairwise-parallel
          minimal separators in $g$, and $h=g_{[\sepset]}$.
	\end{enumerate}
\end{citedtheorem}

Theorem~\ref{thm:ParraS97}, combined with
Theorem~\ref{thm:minsep-chordal-ptime},
gives the desired reduction in the following corollary. Recall that
the graph $\G\ms(g)$ is defined in Section~\ref{sec:msgraph}, as part
of the SGR $\sepsys=(\G\ms,\Av\ms,\Ae\ms)$.

\begin{corollary}\label{cor:triang-to-minsep}
  For a graph $g$, there is a polynomial-time-computable bijection
  between the following two sets:
  \begin{itemize}
  \item $\maxind(\G\ms(g))$, that is, the set of all maximal sets of
    pairwise-parallel minimal separators of $g$.
  \item $\mintri(g)$, that is, the set of all minimal triangulations
    of $g$.
  \end{itemize}
\end{corollary}

Hence, it suffices to prove that $\sepsys$ has a tractable expansion,
which we do next.

\subsection{Enumerating Minimal Separators}\label{sec:min-sep-enum}

{\def\myrulewidth{3in}
  \begin{algseries}{t}{Enumerating $\minsep(g)$ with polynomial
      delay\label{alg:min-sep-enum} (a variation of the algorithm by
      Berry et al.~\cite{conf/wg/BerryBC99})}
		\begin{insidealg}{PDelayAllMinSep}{$g$}
			\STATE $\Q\asn\emptyset$
			\STATE $\MO\asn\emptyset$
			\FORALL{$v\in{\nodes(g)}$}
			\FORALL{$C\in\mathscr{C}(\set{v}\cup{N(v)})$}
			\STATE $\Q\asn\Q\cup\set{N(C)}$
			\ENDFOR
			\ENDFOR
			\WHILE{$\Q\neq\emptyset$}
			\STATE $S\asn\Q.\algname{pop}()$
			\FORALL{$x\in{S}$}
			\STATE $S'\asn\set{N(C)\mid C\in\mathscr{C}(S\cup{N(x)})}$
			\IF{$S'\notin\MO$}
			\STATE $\Q\asn\Q\cup\set{S'}$
			\ENDIF
			\ENDFOR
			\STATE $\MO\asn\MO\cup\set{S}$
			\STATE \textbf{print} $S$
			\ENDWHILE
		\end{insidealg}
\end{algseries}}

We now describe a variation of the algorithm of Berry et
al.~\cite{conf/wg/BerryBC99} that, given a graph $g$, enumerates its
set $\minsep(g)$ of minimal separators.  Their algorithm enumerates
with polynomial total time, and with a simple change (that we explain
next) can enumerate with polynomial delay.  Our variation is depicted
in Figure~\ref{alg:min-sep-enum}. There, for $v\in\nodes(g)$ we denote
by $N(v)$ the set of neighbors of $v$. For $U\subseteq\nodes(g)$ we
denote by $N(U)$ the set of neighbors of nodes in $U$, excluding the
nodes of $U$ themselves; that is,
\[N(U)\eqdef \Big(\bigcup_{v\in U}N(v)\Big)\setminus U\,.\]
We also denote by $\mathscr{C}(U)$ the set of connected components of
the graph $g\sm U$ (the graph obtained from $g$ by removing all the
nodes of $U$).

The algorithm remains intrinsically the same as that of Berry et
al.~\cite{conf/wg/BerryBC99}.  Minimal separators are considered as
neighborhoods of connected components. The algorithm finds minimal
separators contained in a set $U\subseteq\nodes(g)$ by taking the
neighborhoods of the connected components of $g\sm{U}$, that is,
$N(C)$ for all $C\in\mathscr{C}(U)$. Initially, the minimal separators
that are contained in the neighborhoods of single nodes are generated
(lines 3--5). Then, every previously generated minimal separator $S$
is processed to produce more minimal separators that are \e{close} to
$S$ (lines 7--12). For every node $v$ in the minimal separator $S$, it
produces minimal separators that are contained in $S\cup{N(v)}$.

Our modification is in the data structures and the time of printing
answers. In Figure~\ref{alg:min-sep-enum}, $\Q$ and $\MO$ play the
role of $\mathscr{S}\sm\mathscr{T}$ and $\mathscr{T}$ of the original
algorithm~\cite{conf/wg/BerryBC99}, respectively.  There,
$\mathscr{S}$ holds all minimal separators generated, and
$\mathscr{T}$ is a subset that holds the separators that were
processed.  The easy access to the separators yet to be processed
(i.e.  $\mathscr{S}\sm\mathscr{T}$), along with printing answers when
processed (in line 13, rather then when revealed in line 11), provides
the polynomial delay.  Correctness is derived directly by the
correctness of the original algorithm, and the polynomial delay can be
easily verified. In particular, the time between two consecutive
results is $O(|\nodes(g)|^3)$.

\subsection{Tractable Expansion}
Recall that Rose~\cite{Rose1970597} proved that a chordal graph has
fewer minimal separators than nodes.  Combined with this result,
Theorem~\ref{thm:ParraS97} gives the first of the two conditions
of Definition~\ref{def:tractable-expansion}.

\begin{corollary}\label{cor:MISsize}
  Let $g$ be a graph. If $I$ is a (maximal) independent set of
  $\G\ms(g)$, then $|I|<\nodes(g)$.
\end{corollary}
\begin{qproof}
  Suppose that $I$ is a maximal set of pairwise-parallel minimal
  separators of $g$. Then by Theorem~\ref{thm:ParraS97}, $h=g_{[I]}$ is
  a minimal triangulation of $g$, and $\minsep(h)=I$. The graph $h$ is
  chordal, hence from Rose~\cite{Rose1970597} we get that
  $|\minsep(h)|<|V(h)|=|V(g)|$.
\end{qproof}

We now turn to proving the second condition of
Definition~\ref{def:tractable-expansion}. We do so by describing
a general procedure for extending a set of pairwise-parallel minimal
separators of a graph $g$ to a maximal such set. Algorithm
\algname{Extend} of Figure~\ref{alg:BBExtend} can apply any
known polynomial time triangulation heuristic, referred to as
$\trihrs$, as a black box. It uses the following procedures as
subroutines.
\begin{itemize}
\item $\algname{Saturate}(g,S)$ receives a graph $g$ and a set $S
  \subseteq V(g)$ of vertices, and saturates $S$ (i.e., modifies $g$
  such that $S$ becomes a clique).
\item $\trihrs(g)$ receives a graph $g$ and returns a (not necessarily
  minimal) triangulation $g'$ of $g$. We assume that this procedure
  runs in polynomial time. (For example, a naive implementation would
  be to add every possible edge; later we discuss smarter
  alternatives.)

\item $\algname{MinTriSandwich}(g,g')$ receives a graph $g$ and a
  triangulation $g'$ of $g$, and returns a \e{minimal} triangulation of
  $g$. We note that, using one of the known
  algorithms~\cite{Dahlhaus1997,DBLP:journals/siammax/Peyton01,Blair2001125},
  this procedure runs in time that is polynomial in the size of the
  graph.
\item $\algname{ExtractMinSeps}(h)$ receives a chordal graph $h$ and
  returns its set of minimal separators. Using the algorithm of
  Kumar~\cite{Kumar1998155}, the execution time of this procedure is
  linear in $h$.
\end{itemize}

\algname{Extend} takes as input a graph $g$ and a set $\varphi$ of
pairwise-parallel minimal separators.  It then proceeds by saturating
the separators in $\varphi$, resulting in $g_{[\varphi]}$.  At this
stage it passes $g_{[\varphi]}$ to the triangulation heuristic
$\trihrs$.  We note that $\trihrs$ does not have to produce a minimal
triangulation. This is important since it allows us to incorporate
\emph{any} method for triangulation or tree decomposition.  (We
discuss in detail the translation between triangulations and tree
decompositions in Section~\ref{sec:proper}.)

The problem of transforming a non-minimal triangulation into a minimal
one is called the \e{minimal triangulation sandwich
  problem}~\cite{Heggernes2006297}. Various polynomial-time algorithms
for this problem
exist~\cite{Dahlhaus1997,DBLP:journals/siammax/Peyton01}, and these
were reported to perform well in practice~\cite{Blair2001125}.

So, at this stage we have a minimal triangulation $h$ of
$g_{[\varphi]}$. Theorem~\ref{thm:HeggernesSaturation} (that we give
in the next section) shows that $h$ is also a minimal triangulation
of $g$. Lemma~\ref{lemma:Extension} (also in the next section) shows
that the set of minimal separators of $h$ contains $\varphi$, which
is essential as we need to \e{extend} $\varphi$.  Finally, we can
apply the algorithm of Kumar~\cite{Kumar1998155} to extract the
minimal separators of the (chordal) graph $h$ in linear time.

{\def\myrulewidth{3in}
  \begin{algseries}{t}{An algorithm for extending a set $\varphi$ of
      pairwise-parallel minimal separators \label{alg:BBExtend}}
    \begin{insidealg}{Extend}{$g$,$\varphi$}      
      \STATE $g_t \asn \trihrs(g_{[\varphi]})$ \label{algline:triHrs}
      \STATE $h \asn \algname{MinTriSandwich}(g_{[\varphi]},g_t)$ \label{algline:sandwich}
      \STATE \textbf{return} $\algname{ExtractMinSeps}(h)$ \label{algline:ExtractSeps}
		\end{insidealg}
	\end{algseries}}

     \subsubsection{Correctness.} To prove correctness of the
      algorithm $\algname{Extend}$ of Figure~\ref{alg:BBExtend}, we
      need the following result by Heggernes~\cite{Heggernes2006297}.

      \begin{citedtheorem}{Heggernes~\cite{Heggernes2006297}}\label{thm:HeggernesSaturation}
        Given a graph $g$, let $\varphi$ be an arbitrary set of
        pairwise-parallel minimal separators of $g$. Obtain a
        graph $g_{[\varphi]}$ by saturating each separator in
        $\varphi$.
        \begin{enumerate}
        \item \label{item:Heggernes1st} $\varphi \subseteq
          \clqminsep(g_{[\varphi]})$, that is, $\varphi$ consists of
          clique minimal separators of $g_{[\varphi]}$.
        \item \label{item:Heggernes2nd} $\clqminsep(g) \subseteq
          \minsep(g_{[\varphi]})$; that is, every clique minimal
          separator of $g$ is a (clique) minimal separator of
          $g_{[\varphi]}$.
	\item \label{item:Heggernes4th} Every minimal triangulation of
          $g_{[\varphi]}$ is a minimal triangulation of $g$.
\end{enumerate}
\end{citedtheorem}

The next lemma builds on Theorems~\ref{thm:ParraS97}
and~\ref{thm:HeggernesSaturation}.

\def\lemmaextension{
 Let $g$ be a graph, and $\varphi$ a set of pairwise-parallel minimal
  separators of $g$. Let $h$ be a minimal triangulation of
  $g_{[\varphi]}$. Then $\varphi \subseteq \minsep(h)$.
}

\begin{lemma}\label{lemma:Extension}
  \lemmaextension
\end{lemma}
\begin{proof}
  By Part~\ref{item:Heggernes1st} of
  Theorem~\ref{thm:HeggernesSaturation} we have that $\varphi
  \subseteq \clqminsep(g_{[\varphi]})$.  Since $h$ is a minimal
  triangulation of $g_{[\varphi]}$ then by Part~\ref{item:ParraS2nd} of Theorem~\ref{thm:ParraS97},
  $h$ is the result of saturating a maximal set, say $\varphi'$, of
  pairwise-parallel minimal separators of $g_{[\varphi]}$.  Therefore,
  by Part~\ref{item:Heggernes2nd} of
  Theorem~\ref{thm:HeggernesSaturation} we have
  $\clqminsep(g_{[\varphi]}) \subseteq \minsep(h)$.  This implies
  that $\varphi \subseteq \minsep(h)$, as claimed.
\end{proof}

We then conclude the correctness of the algorithm.

\def\lemmabbcorrectness{
  Let $\varphi$ be a set of pairwise-parallel minimal separators of a
  graph $g$. \algname{Extend}$(g,\varphi)$ returns a maximal set
  $I$ of pairwise-parallel minimal separators of $g$ such that $\varphi \subseteq
  I$. Furthermore, the algorithm terminates in polynomial time.
}

\begin{lemma}\label{lemma:BBCorrectness}
\lemmabbcorrectness
\end{lemma}
\begin{proof}
  Assuming correctness of procedures $\trihrs$, and
  $\algname{MinTriSandwich}$, the graph $h$ is a minimal triangulation
  of $g_{[\varphi]}$. By Part~\ref{item:Heggernes4th} of
  Theorem~\ref{thm:HeggernesSaturation}, we have that $h$ is a minimal
  triangulation of $g$. Consequently, from Part~\ref{item:ParraS2nd} of Theorem~\ref{thm:ParraS97}
  we get that $\minsep(h)=I$ is a maximal set of pairwise-parallel
  minimal separators of $g$. By Lemma~\ref{lemma:Extension} it holds
  that $\varphi \subseteq \minsep(h)$, making $I$ an extension of
  $\varphi$.  All of the procedures in Figure~\ref{alg:BBExtend} run
  in time that is polynomial in the size of the graph making it
  polynomial as well.
\end{proof}

From Corollary~\ref{cor:MISsize} and Lemma~\ref{lemma:BBCorrectness}
we get the main result of this part.

\begin{theorem}\label{thm:poly-exp}
  The SGR $\sepsys$ has a tractable expansion of independent sets.
\end{theorem}
This theorem allows us to establish the main result of this paper.

\subsection{Main Result}

From Theorems~\ref{thm:poly-exp} and~\ref{thm:sgr-inc-ptime} we
conclude that it is possible to enumerate the maximal independent sets
of $\sepsys$ in incremental polynomial time. Applying the bijection of
Corollary~\ref{cor:triang-to-minsep}, we get the main result of this
paper.

\begin{corollary}\label{cor:min-triang}
  Given a graph, the minimal triangulations can be enumerated in
  incremental polynomial time.
\end{corollary}

In the next section, we will use this result for enumerating tree
decompositions.

\eat{
From Theorems~\ref{thm:poly-exp} and~\ref{thm:sgr-inc-ptime} we
conclude that it is possible to enumerate the maximal independent sets
of $\sepsys$ in incremental polynomial time. Applying the bijections
of Theorems~\ref{cor:triang-to-minsep} and~\ref{thm:proper-mint}, we
get the main result of this paper.

\begin{theorem}\label{thm:main}
  There are algorithms that, given a graph $g$, enumerate in
  incremental polynomial time:
  \begin{enumerate}
\item The minimal triangulations of $g$.
\item The proper tree decompositions of $g$.
\end{enumerate}
\end{theorem}

}

\eat{
\begin{theorem}
$\sepsys=(\G,\Av,\Ae)$ is polynomial.
\end{theorem}
\begin{qproof}
  The algorithm $\Av$ that enumerates the minimal separators of
  $\G(g)$ runs in polynomial delay, within time $O(n^3)$ between
  consecutive minimal separators. Since the string representation of a
  graph $g$ is at least as large as the cardinality of its vertex set,
  the first requirement is met.

By the definition of separator graph, testing for the existence of an
edge $(u,v) \in \sigma_g$ is equivalent to testing whether $u \crosses
v$. Algorithm $\Ae$ will run a depth first search over the graph $g
\sm u$ ($u$ represents a minimal vertex separator), beginning from
some vertex in the separator $v$. The algorithm returns true if not
all the members of $v$ are part of the resulting DFS tree (belong to
the same connected component). This can be done in linear time.
\end{qproof}
}

\section{Enumerating the Proper Tree Decompositions}\label{sec:proper}
In this section we define the notion of a \e{proper} tree
decomposition, which is essentially a tree decomposition that is,
intuitively, not deemed redundant due to another tree
decomposition. Our ultimate goal is to enumerate \e{only} the proper
tree decompositions, and we will show that this translates to
enumerating the minimal triangulations.

\subsection{Proper Tree Decompositions}

Let $d_1$ and $d_2$ be two tree decompositions of a graph $g$.  We say
that $d_1$ and $d_2$ are \e{bag equivalent}, denoted $d_1\equivb d_2$,
if $\bags(d_1)=\bags(d_2)$.  We denote by $d_1\sqsubseteq d_2$ the
fact that for every bag $b_1\in\bags(d_1)$ there exists a bag
$b_2\in\bags(d_2)$ such that $b_1\subseteq b_2$.

Let $g$ be a graph, and let $d$ and $d'$ be tree decompositions of
$g$. We say that $d'$ \e{strictly subsumes} $d$ if $d'\sqsubseteq d$
and $\bags(d)\not\subseteq\bags(d')$ in multiset notation (i.e., some bag appears in $d$ more times than it appears in $d'$). A tree decomposition is
\e{proper} if it is not strictly subsumed by any tree decomposition,
and it is \e{improper} otherwise. 

Figure \ref{fig:decompositions} shows examples of proper and improper
tree decompositions. It can be shown that $d_1$ is proper (e.g., since
every clique of $g$ is contained in some bag of $d$, as we prove in
Proposition~\ref{proposition:cliquesInBags}). But
$d_2$ is not proper, since it is subsumed by $d_1$; that is, every bag
of $d_1$ is contained in some bag of $d_2$, but the bag
$\set{1,2,3,4}$ is not a bag of $d_1$. For the same reason, $d_2$ is
subsumed also by $d_3$. Finally, $d_3$ is subsumed by $d_1$ since
every bag of $d_1$ is a bag of $d_3$, but the bag $\set{3,4}$ is not a
bag of $d_1$.

\begin{figure}[t]
\centering
\begin{picture}(0,0)%
\includegraphics{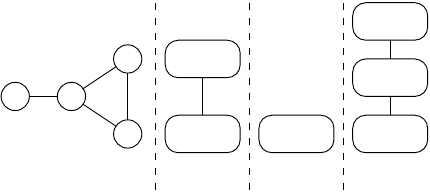}%
\end{picture}%
\setlength{\unitlength}{3947sp}%
\begingroup\makeatletter\ifx\SetFigFont\undefined%
\gdef\SetFigFont#1#2#3#4#5{%
  \reset@font\fontsize{#1}{#2pt}%
  \fontfamily{#3}\fontseries{#4}\fontshape{#5}%
  \selectfont}%
\fi\endgroup%
\begin{picture}(3433,1524)(1755,-2173)
\put(3376,-1786){\makebox(0,0)[b]{\smash{{\SetFigFont{9}{10.8}{\familydefault}{\mddefault}{\updefault}{\color[rgb]{0,0,0}$2,3,4$}%
}}}}
\put(3376,-1186){\makebox(0,0)[b]{\smash{{\SetFigFont{9}{10.8}{\familydefault}{\mddefault}{\updefault}{\color[rgb]{0,0,0}$1,2$}%
}}}}
\put(3376,-2086){\makebox(0,0)[b]{\smash{{\SetFigFont{9}{10.8}{\familydefault}{\mddefault}{\updefault}{\color[rgb]{0,0,0}$d_1$}%
}}}}
\put(4876,-1336){\makebox(0,0)[b]{\smash{{\SetFigFont{9}{10.8}{\familydefault}{\mddefault}{\updefault}{\color[rgb]{0,0,0}$2,3,4$}%
}}}}
\put(4876,-2086){\makebox(0,0)[b]{\smash{{\SetFigFont{9}{10.8}{\familydefault}{\mddefault}{\updefault}{\color[rgb]{0,0,0}$d_3$}%
}}}}
\put(4876,-1786){\makebox(0,0)[b]{\smash{{\SetFigFont{9}{10.8}{\familydefault}{\mddefault}{\updefault}{\color[rgb]{0,0,0}$3,4$}%
}}}}
\put(4876,-886){\makebox(0,0)[b]{\smash{{\SetFigFont{9}{10.8}{\familydefault}{\mddefault}{\updefault}{\color[rgb]{0,0,0}$1,2$}%
}}}}
\put(2326,-2086){\makebox(0,0)[b]{\smash{{\SetFigFont{9}{10.8}{\familydefault}{\mddefault}{\updefault}{\color[rgb]{0,0,0}$g$}%
}}}}
\put(4126,-1786){\makebox(0,0)[b]{\smash{{\SetFigFont{9}{10.8}{\familydefault}{\mddefault}{\updefault}{\color[rgb]{0,0,0}$1,2,3,4$}%
}}}}
\put(4126,-2086){\makebox(0,0)[b]{\smash{{\SetFigFont{9}{10.8}{\familydefault}{\mddefault}{\updefault}{\color[rgb]{0,0,0}$d_2$}%
}}}}
\put(1876,-1448){\makebox(0,0)[b]{\smash{{\SetFigFont{9}{10.8}{\familydefault}{\mddefault}{\updefault}{\color[rgb]{0,0,0}1}%
}}}}
\put(2326,-1448){\makebox(0,0)[b]{\smash{{\SetFigFont{9}{10.8}{\familydefault}{\mddefault}{\updefault}{\color[rgb]{0,0,0}2}%
}}}}
\put(2776,-1748){\makebox(0,0)[b]{\smash{{\SetFigFont{9}{10.8}{\familydefault}{\mddefault}{\updefault}{\color[rgb]{0,0,0}4}%
}}}}
\put(2776,-1148){\makebox(0,0)[b]{\smash{{\SetFigFont{9}{10.8}{\familydefault}{\mddefault}{\updefault}{\color[rgb]{0,0,0}3}%
}}}}
\end{picture}%
 \caption{A graph $g$ and tree decompositions $d_1$, $d_2$ and $d_3$ of
  $g$. The decomposition $d_1$ is proper, but $d_2$ and $d_3$ are
  subsumed by $d_1$, and hence, improper.}
  \label{fig:decompositions}
\end{figure}

\subsection{Enumeration}

The main result of this section is the following, showing that
enumerating the proper tree decompositions reduces to enumerating the
minimal triangulations.

\begin{theorem}\label{thm:proper-mint}
  Let $g$ be a graph. There is a bijection $M$ between $\mintri(g)$
  and the equivalence classes of $\equivb$ over the proper tree
  decompositions of $g$. Moreover, given a minimal triangulation $h$
  of $g$, the proper tree decompositions in the class $M(h)$ can be
  enumerated with polynomial delay.
\end{theorem}

Combined with Corollary~\ref{cor:min-triang}, we get the following.

\begin{corollary}\label{cor:proper-td}
  The set of proper tree decompositions of a given graph can be
  enumerated in incremental polynomial time.
\end{corollary}

Next, we discuss the proof of Theorem~\ref{thm:proper-mint}, and in
particular show how $M$ is defined. 
We first
need some propositions.  The following proposition is a folklore, and
it is using a result by Heggernes~\cite{Heggernes:2006}, showing that
every collection of subtrees of a tree satisfies the \e{Helly
  property}.

\def\propositioncliquesInBags{ If $d$ is a tree decomposition of a
  graph $g$, then every clique of $g$ is contained in some bag of $d$.
}
\begin{proposition}\label{proposition:cliquesInBags}
 \propositioncliquesInBags
\end{proposition}
\begin{proof}
  We use the fact that the junction-tree property of a tree
  decomposition is equivalent to the property that for every node $v$
  of the graph, the bags of the tree decomposition that contain $v$
  form a (connected) subtree. Denote $d=(t,\beta)$ and let $C$ be a clique of $g$.
  Every node $v$ in $C$ defines a subtree of $t$ that is induced by
  the bags that contain $v$. Since $d$ covers the edges of $g$, every
  two nodes in $C$ must share some bag in $d$, and hence, their
  subtrees must share a vertex.  Heggernes~\cite{Heggernes:2006} shows
  that every collection of subtrees of a tree satisfies the \e{Helly
    property}: if every two subtrees share a vertex, then there exists
  a vertex that is shared by all the subtrees.  In particular, there
  exists a vertex in $d$ common to all of these subtrees; this shared
  node corresponds to a bag that contains $C$.
\end{proof}

The following proposition states that in a proper tree decomposition,
there is no containment among bags.

\begin{figure}[t]
	\centering
\begin{picture}(0,0)%
\includegraphics{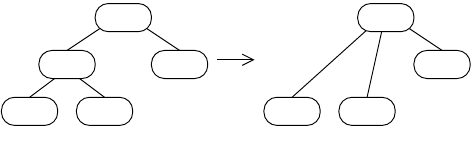}%
\end{picture}%
\setlength{\unitlength}{3947sp}%
\begingroup\makeatletter\ifx\SetFigFont\undefined%
\gdef\SetFigFont#1#2#3#4#5{%
  \reset@font\fontsize{#1}{#2pt}%
  \fontfamily{#3}\fontseries{#4}\fontshape{#5}%
  \selectfont}%
\fi\endgroup%
\begin{picture}(3774,1267)(1339,-1991)
\put(4426,-1936){\makebox(0,0)[b]{\smash{{\SetFigFont{9}{10.8}{\familydefault}{\mddefault}{\updefault}{\color[rgb]{0,0,0}$d'$}%
}}}}
\put(2326,-1936){\makebox(0,0)[b]{\smash{{\SetFigFont{9}{10.8}{\familydefault}{\mddefault}{\updefault}{\color[rgb]{0,0,0}$d$}%
}}}}
\put(2326,-886){\makebox(0,0)[b]{\smash{{\SetFigFont{9}{10.8}{\familydefault}{\mddefault}{\updefault}{\color[rgb]{0,0,0}$A$}%
}}}}
\put(3676,-1636){\makebox(0,0)[b]{\smash{{\SetFigFont{9}{10.8}{\familydefault}{\mddefault}{\updefault}{\color[rgb]{0,0,0}$D$}%
}}}}
\put(4426,-886){\makebox(0,0)[b]{\smash{{\SetFigFont{9}{10.8}{\familydefault}{\mddefault}{\updefault}{\color[rgb]{0,0,0}$A$}%
}}}}
\put(1876,-1261){\makebox(0,0)[b]{\smash{{\SetFigFont{9}{10.8}{\familydefault}{\mddefault}{\updefault}{\color[rgb]{0,0,0}$B$}%
}}}}
\put(2776,-1261){\makebox(0,0)[b]{\smash{{\SetFigFont{9}{10.8}{\familydefault}{\mddefault}{\updefault}{\color[rgb]{0,0,0}$C$}%
}}}}
\put(2176,-1636){\makebox(0,0)[b]{\smash{{\SetFigFont{9}{10.8}{\familydefault}{\mddefault}{\updefault}{\color[rgb]{0,0,0}$E$}%
}}}}
\put(1576,-1636){\makebox(0,0)[b]{\smash{{\SetFigFont{9}{10.8}{\familydefault}{\mddefault}{\updefault}{\color[rgb]{0,0,0}$D$}%
}}}}
\put(4876,-1261){\makebox(0,0)[b]{\smash{{\SetFigFont{9}{10.8}{\familydefault}{\mddefault}{\updefault}{\color[rgb]{0,0,0}$C$}%
}}}}
\put(4276,-1636){\makebox(0,0)[b]{\smash{{\SetFigFont{9}{10.8}{\familydefault}{\mddefault}{\updefault}{\color[rgb]{0,0,0}$E$}%
}}}}
\end{picture}%
 	\caption{Obtaining a strictly subsuming tree decomposition $d'$ given
		a tree decomposition $d$ with $B \subseteq C$.}
	\label{fig:antichain}
\end{figure}

\def\propositionantichain{
If $d$ is a proper tree decomposition of a graph $g$, then $\bags(d)$
is an antichain w.r.t.~set inclusion (that is, no bag contains
another).
}
\begin{proposition}
\label{proposition:antichain}
\propositionantichain
\end{proposition}
\begin{qproof}
  We need to show that a proper tree decomposition cannot have two
  bags with one contained in the other.  Assume, by way of
  contradiction, that $d$ is a proper tree decomposition of $g$ with two bags
  $B,C \in \bags(d)$ where $B \subseteq C$.  Let $A$ be the second bag
  in the path from $B$ to $C$. Since $d$ is a tree decomposition and
  $A$ is on the path from $B$ to $C$, we get that $B=B \cap C
  \subseteq A$.

  Define $d'$ to be the graph obtained from $d$ by removing $B$ and
  connecting $A$ to all other neighbors of $B$, as illustrated in
  Figure~\ref{fig:antichain}. We will show that $d'$ is a tree
  decomposition for $g$.  The first two properties of the tree
  decomposition still hold because $A$ contains $B$. Consider the path
  between two bags $\alpha$ and $\beta$ of $d'$. If the path between
  them is the same as in $d$, the third property still holds. If it
  changed, then the path used to go through $B$, and the only new bag
  that may appear in this path is $A$. In this case, $\alpha \cap
  \beta \subseteq B \subseteq A$, and the third property holds as
  well.  We have found a tree decomposition $d'$ for $g$ that strictly
  subsumes $d$, hence $d$ is improper, and this is a contradiction.
\end{qproof}
 
From Theorem~\ref{theorem:Jordan}, the following easily follows.
\def\propositionsat{
  If $d$ is a tree decomposition of a graph $g$, then $\sat(g,d)$ is a
  triangulation of $g$.
}

\begin{proposition}\label{proposition:sat}
\propositionsat
\end{proposition}
\begin{proof}
  According to the definitions, $d$ is a tree decomposition of $\sat(g,d)$. Hence, since every
  bag of $d$ is a clique of $\sat(g,d)$, it follows from
  Theorem~\ref{theorem:Jordan} that $\sat(g,d)$ is chordal.
\end{proof}

The definition of $M$ is based on
Lemma~\ref{lemma:chordal-one-td}, stating that a chordal graph $g$ has a single
proper tree decomposition, up to the equivalence $\equivb$, with the
set of bags being precisely the set of maximal cliques.

\begin{lemma}\label{lemma:chordal-one-td}
If $g$ is a chordal graph and $d$ is a proper
  tree decomposition of $g$, then $\bags(d)=\maxcl(g)$.
\end{lemma}

\begin{proof}
  According to Proposition~\ref{proposition:cliquesInBags}, every
  clique of $g$ is contained in some bag of $d$, and according to
  Theorem~\ref{theorem:Jordan}, $g$ has some tree decomposition, say
  $d'$, where all the bags are cliques of $g$. So we have that
  $d'\sqsubseteq d$.  If $\bags(d)\not\subseteq\bags(d')$, then $d'$
  strictly subsumes $d$, in contradiction to the fact that $d$ is
  proper. Hence $\bags(d)\subseteq\bags(d')$, meaning that the bags of
  $d$ are cliques of $g$. It thus follows that every maximal clique is
  a bag of $d$, or in notation, $\maxcl(g)\subseteq\bags(d)$.
  Finally, Proposition~\ref{proposition:antichain} states that the
  bags of $d$ are an antichain w.r.t.~set inclusion, and hence,
  $\bags(d)\subseteq\maxcl(g)$. We conclude that $\bags(d)=\maxcl(g)$,
  as claimed.
\end{proof}

Based on Lemma~\ref{lemma:chordal-one-td}, we define $M$ to be the
function that maps every $h\in \mintri(g)$ to the equivalence class of
the proper tree decomposition of $h$. 
Lemma~\ref{lemma:proper-correct} states that $M$ has the required properties.

\begin{lemma}\label{lemma:proper-correct}
 Let $g$ be a graph. The mapping $M$ is a bijection between
  $\mintri(g)$ and the equivalence classes of $\equivb$ over the
  proper tree decompositions of $g$.
\end{lemma}

\begin{proof}
  We show that $M$ has the correct range, that it is surjective,
  and that it is bijective.

  \partitle{$M$ has a proper range}  
  Let $h$ be a minimal triangulation of $g$, and let $d$ be a proper
  tree decomposition of $h$ in $M(h)$.  Then $d$ is also a tree
  decomposition of $g$, as the three properties of a tree
  decomposition still hold. We need to show that $d$ is a \e{proper}
  tree decomposition of $g$. According to
  Lemma~\ref{lemma:chordal-one-td}, we have $\bags(d)=\maxcl(h)$, and
  therefore, $\sat(g,d)=h$.  Assume, by way of contradiction, that $d$
  is improper. Then $d$ is strictly subsumed by some tree
  decomposition $d'$ of $g$, meaning that $d'\sqsubseteq d$.  Let $h'$ be the graph
  $\sat(g,d')$. From Proposition~\ref{proposition:sat} it follows that
  $h'$ is a triangulation of $g$. From $d'\sqsubseteq d$ and the fact
  that every bag of $d$ is a clique of $h$, we conclude that
  $\edges(h')\subseteq\edges(h)$. And since $h$ is a minimal
  triangulation, we get that $h=h'$. We can now conclude that $d'$ is
  also a tree decomposition of $g$: the junction-tree property holds
  and the nodes are covered since it is a tree decomposition of $g$,
  and the edges are covered since those are the edges of $h'$ that are
  covered by its definition.  We get that both $d$ and $d'$ are tree
  decompositions of $h$, and $d$ is strictly subsumed by $d'$, which
  contradicts the fact that $d$ is a proper tree decomposition of $h$.
  
\partitle{$M$ is injective}
Let $h_1$ and $h_2$ be two minimal triangulations such that
$h_1 \neq h_2$. Without loss
of generality, assume that the edge $\set{u,v}$ is in $h_1$ but not in
$h_2$. The nodes $u$ and $v$ are part of some maximal clique of $h_1$,
so they share a bag in $M(h_1)$. But they are not part of any clique
of $h_2$, so they do not share any bag in $M(h_2)$. Therefore,
$M(h_1) \neq M(h_2)$.

\partitle{$M$ is surjective} 
Given a proper tree decomposition $d$ of $g$, we need to show that
there exists a minimal triangulation $h$ of $g$ such that $d\in M(h)$.
Consider the graph $h=\sat(g,d)$. We will show that $h$ is a minimal
triangulation, and that $d$ belongs to $M(h)$.

We first show that $h$ is a minimal triangulation of $g$.  According
to Proposition~\ref{proposition:sat}, $h$ is a triangulation of $g$.
Assume, by way of contradiction, that $h$ is not minimal.  Then there
exists a minimal triangulation $h'$ of $g$ that is obtained from $h$
by removing some edges; denote one of these edges by $e$.  Consider a
tree decomposition $d'\in M(h')$. The clique containing $e$ in $h$ is
not a clique in $h'$, and therefore $\bags(d)\not\subseteq\bags(d')$.
Also note that since $h' \subseteq h$, every maximal clique of $h'$ is
contained in some maximal clique of $h$, and therefore
$d'\sqsubseteq d$.  Then $d'$ strictly subsumes $d$, in contradiction
to the fact that $d$ is proper.

Finally, we need to show that $d$ is a proper tree decomposition of
$h$. The nodes of $h$ are covered in $d$, and the junction-tree
property holds, since $d$ is a tree decomposition of $g$. The new
edges of $h$ are covered in $d$ since they are all a result of
saturation of the bags of $d$. So $d$ is a tree decomposition of $h$,
and we claim that it is proper. Assume, by way of contradiction, that
$d$ is not a proper tree decomposition of $h$, then the tree
decomposition $d'$ that strictly subsumes it is also a tree
decomposition for $g$, contradicting the fact that $d$ is a proper
tree decomposition of $g$.
\end{proof}

To complete the proof of Theorem~\ref{thm:proper-mint}, we explain how
the proper tree decompositions in the class $M(h)$ can be enumerated
with polynomial delay for $h\in\mintri(g)$. Jordan~\cite{Jordan} shows
that, given a chordal graph $h$, a tree over the bags that represent
the maximal cliques of $h$ is a tree decomposition if and only if it
is a maximal spanning tree, where the weight of an edge between two
bags is the size of their intersection. Hence, this enumeration
problem is reduced to enumerating all maximal spanning trees, which
can be solved in polynomial
delay~\cite{DBLP:journals/ijcm/YamadaKW10}.  Since Gavril~\cite{GAVR}
showed that in chordal graphs the number of maximal cliques of $h$
is at most the number of nodes of $h$, we have a polynomial delay
algorithm for enumerating the tree decompositions.  This concludes the
proof.

\addedn{According to Corollary~\ref{cor:proper-td}, we can enumerate all proper tree decompositions with incremental polynomial time. Note that this section also implies another alternative: we can enumerate only one representative of every equivalence class with the same complexity guarantees. That is, we can enumerate one proper tree-decomposition of each possible bag configuration with incremental polynomial time. The choice of which variation to use depends on the application at hand. For some applications, different tree-decompositions with the same bags may be of different quality, while for others only the bags matter.}

\section{Experimental Evaluation}\label{sec:exper}

We now describe an experimental study over an implementation of our
enumeration algorithm for minimal triangulations, namely the algorithm
\algname{EnumMIS} of Figure~\ref{alg:EnumMaxIndependent} for the SGR
$(\G\ms,\Av\ms,\Ae\ms)$, calling the procedure $\Extend$ of
Figure~\ref{alg:BBExtend}. The goal of the experimental study is
twofold. First, we wish to understand how practical the execution cost
of the algorithm is for enumerating many minimal triangulations (and
tree decompositions). Second, we wish to study the ability of the
algorithm to produce many \e{high-quality} triangulations, given an
underlying triangulation algorithm (for $\Extend$), and even to
improve upon standard quality measures of the underlying algorithm
itself. In Section~\ref{sec:ExperimentalSetup} we describe the
experimental setup, in Section~\ref{sec:Delay} we report on the
efficiency of the algorithm in terms of its delay, and in
Sections~\ref{sec:Quality} and~\ref{sec:testcase} we study the quality
of the results.

\definecolor{ObjectDetection}{RGB}{91,155,213}
\definecolor{Segmentation}{RGB}{237,125,49}
\definecolor{Pedigree}{RGB}{165,165,165}
\definecolor{Grids}{RGB}{255,192,0}
\definecolor{Promedas}{RGB}{68,114,196}
\definecolor{CSP}{RGB}{112,173,71}

\subsection{Experimental Setup}\label{sec:ExperimentalSetup}
We first describe the general setup for our study.

\subsubsection{Implementation and Hardware}
We implemented all algorithms in C++, with STL data
structures.\footnote{The code is available online:
  \url{https://github.com/NofarCarmeli/MinTriangulationsEnumeration}}
All experiments were carried out on a 2.6GHz dual-core laptop with 8GB
of RAM running Windows 10 professional.

\subsubsection{Triangulation Algorithms}
We implemented two well known triangulation algorithms as the
procedure $\trihrs$ in line~\ref{algline:triHrs} of the procedure
$\Extend$ (Figure~\ref{alg:BBExtend}). Both algorithms apply the
general technique of \e{node-elimination
  ordering}~\cite{OHTSUKI1976622}, where nodes are eliminated from the
graph in turn, by adding a subset of fill edges between the eliminated
node and its neighbors in the (leftover) graph.  Both algorithms
guarantee a minimal triangulation (hence there was no need to call
$\algname{MinTriSandwich}(g_{[\varphi]},g_t)$ in
line~\ref{algline:sandwich} of $\Extend$).
\begin{itemize}
\item $\MSCM$~\cite{Berry:2002:MCS:647683.732496}. This is an
  extension of the \e{Maximum Cardinality Search} (MCS) algorithm for
  recognizing chordal graphs~\cite{Tarjan:1984:SLA:1169.1179}, which
  finds a minimal elimination ordering along with its corresponding
  minimal triangulation.
\item $\LBTRI$~\cite{BERRY200633}. This algorithms guarantees
  minimality of the triangulation by adding only a subset of the fill
  edges at each of the elimination steps, and allows for complete
  flexibility in determining the elimination order. We applied the
  \e{min fill} heuristic that selects, at each iteration, the node
  whose elimination results in the smallest number of edges to add.
\end{itemize}

\subsubsection{Datasets}
We used three types of datasets: probabilistic graphical models,
database queries, and random (synthetic) graphs.  For the first type,
we used the following benchmark networks from the UAI probabilistic
inference
challenge.\footnote{\small\url{http://www.cs.huji.ac.il/project/PASCAL/showNet.php}}
The datasets Alchemy and DBN from the challenge are not described here
as each of their graphs had only one or two minimal triangulations,
and the enumeration ended instantaneously.
\begin{itemize}
\item \textbf{Promedas}: ``PRObabilistic MEdical Diagnostic Advisory
  System.'' The Promedas Markov networks represent medical diagnosis
  cases, and consist of binary variables that were converted from
  layered noisy-or Bayesian networks. The dataset includes $33$ graphs
  with $26$-$1039$ nodes and $36$-$1696$ edges, and many of them are considered too
  difficult for exact
  inference.\footnote{\small\url{http://graphmod.ics.uci.edu/uai08/Evaluation/Report/Benchmarks}}
\item \textbf{Object detection}: Markov Random Fields for
  object-detection tasks in computer vision. It includes $79$
  instances of connected networks, each containing $60$ nodes and
  between $135$ to $180$ edges.
\item \textbf{Image segmentation}: Bayesian networks generated from
  image-segmentation tasks. It includes $6$ graphs with $226$-$235$
  nodes and $617$-$647$ edges.
\item \textbf{Grids}: An $N\times N$ grid network. Such networks that
  are common in image processing~\cite{Blake:2011:MRF:2024611}, and
  networks that model problems such as medical diagnosis and object
  detection. This dataset includes $8$ grids with $N =10$ and $N=20$, resulting in graphs with $100$ or $400$ nodes, and $180$-$760$ edges.
\item \textbf{Pedigree}: Bayesian networks used to model genetic
  information~\cite{FishelsonG02}. The data set includes $3$ graphs,
  each has $385$ nodes and $930$ edges.
\item \textbf{CSP}: Constraint-satisfaction problems. There are $3$
  instances in the dataset, with $67$-$100$ nodes and $226$-$619$ edges.
\end{itemize}

The datasets of second and third types are as follows.
\begin{itemize}
\item \textbf{TPC-H}: Graphs induced from TPC-H. These are the Gaifman
  (primal) graphs of joins for implementing the TPC-H benchmark
  queries in LogiQL, the Datalog variant of
  LogicBlox~\cite{DBLP:conf/sigmod/ArefCGKOPVW15}.\footnote{The
    queries, provided to us by LogicBlox, are used for benchmarking
    the engine.}  The queries include up to 22 nodes, and up to 46
  edges, and their treewidth is up to 7.
\item \textbf{Random}: Random $G(n,p)$ graphs in the Erd\H{o}s-R\'enyi
  model. The number of nodes is $n$ and every pair of nodes is
  connected by an edge with probability $p$ (independently).  We
  generated $54$ random graphs for varying $n$ between $30$ and $200$,
  and three values of $p$: $0.3$ (sparsest), $0.5$ and $0.7$
  (densest).
\end{itemize}

As a baseline approach, we implemented the algorithm of
DunceCap~\cite{DBLP:conf/sigmod/TuR15} for generating all of the
generalized hypertree decompositions (each involving an underlying
tree decomposition). However, this algorithm is designed to handle
small join queries and to span a much greater space of objects
(namely, the generalized hypertree decompositions). In particular, on
the TPC-H dataset we observed that on the smaller queries our
algorithm is faster by 3 to 4 orders of magnitude, and on some of the
larger queries (\textsf{Q7} and \textsf{Q9}) we could not get their
algorithm to terminate in less than two hours (while our algorithms
terminated in a few seconds, as we later discuss). Therefore, we
decided to exclude comparisons to this implementation.  As of today,
we are not aware of any other published algorithms for enumerating
(minimal) triangulations or tree decompositions with guarantees of
correctness (completeness).

\subsection{Execution Cost}\label{sec:Delay}
In what follows we report on the delay of the two variants of the
implementation, corresponding to the two triangulation algorithms
$\LBTRI$ and $\MSCM$.

\subsubsection{Probabilistic graphical models}
We measured the average delay between minimal triangulation printouts
for the network datasets from the UAI challenge. The measurements were
conducted during $30$ minutes executions. $5$ of the graphs in Promedas, and one graph of CSP completed the enumeration within this time. We plotted the delay of the other graphs against
the number of their edges. The plots, corresponding to the two
minimal triangulation algorithms $\LBTRI$ and $\MSCM$, are presented
in Figures~\ref{fig:BNAvgDelayLB_Triang_Fill}
and~\ref{fig:BNAvgDelayMCS_M}, respectively, \addedn{using log-scale. %
Overall, we see that the delay increases with the size of the graph.}
However, this trend varies between the different benchmarks.
While this dependency is apparent for the Promedas data set, the average delay for object detection has little correlation with the number of edges in the graph.

\begin{figure*}[tbh]
	\label{fig:BNDelay}
	\centering
	\begin{subfigure}[b]{0.9\linewidth}
	{\includegraphics[width=0.9\textwidth]{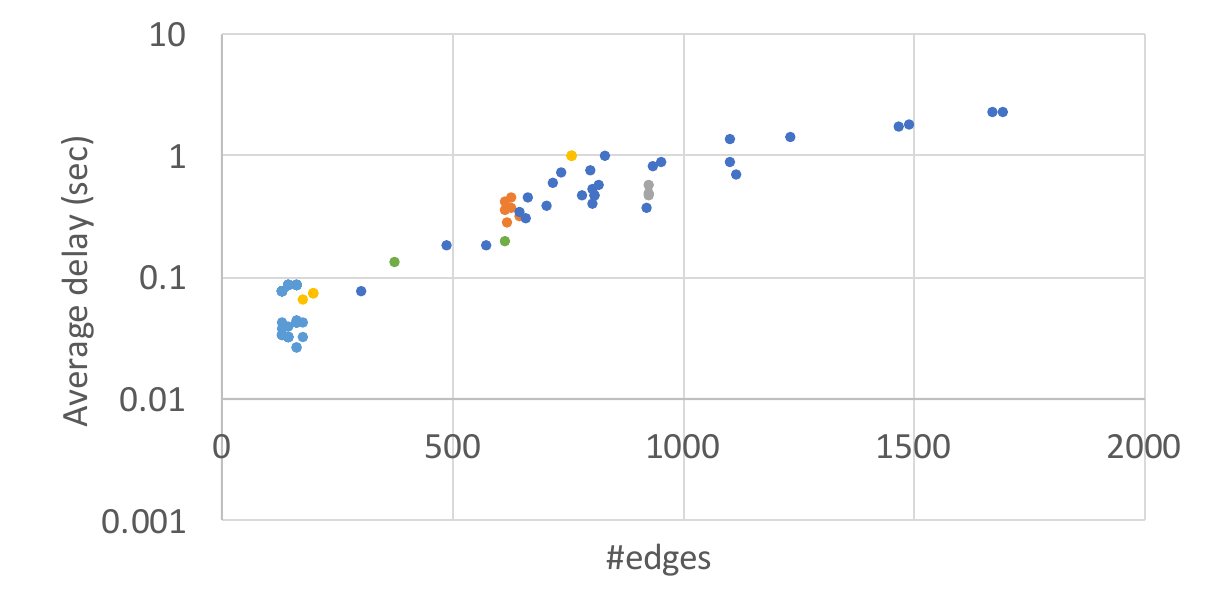}
	\caption{$\LBTRI$} \label{fig:BNAvgDelayLB_Triang_Fill}}
	\end{subfigure}
	\begin{subfigure}[b]{0.9\linewidth}
	{\includegraphics[width=0.9\textwidth]{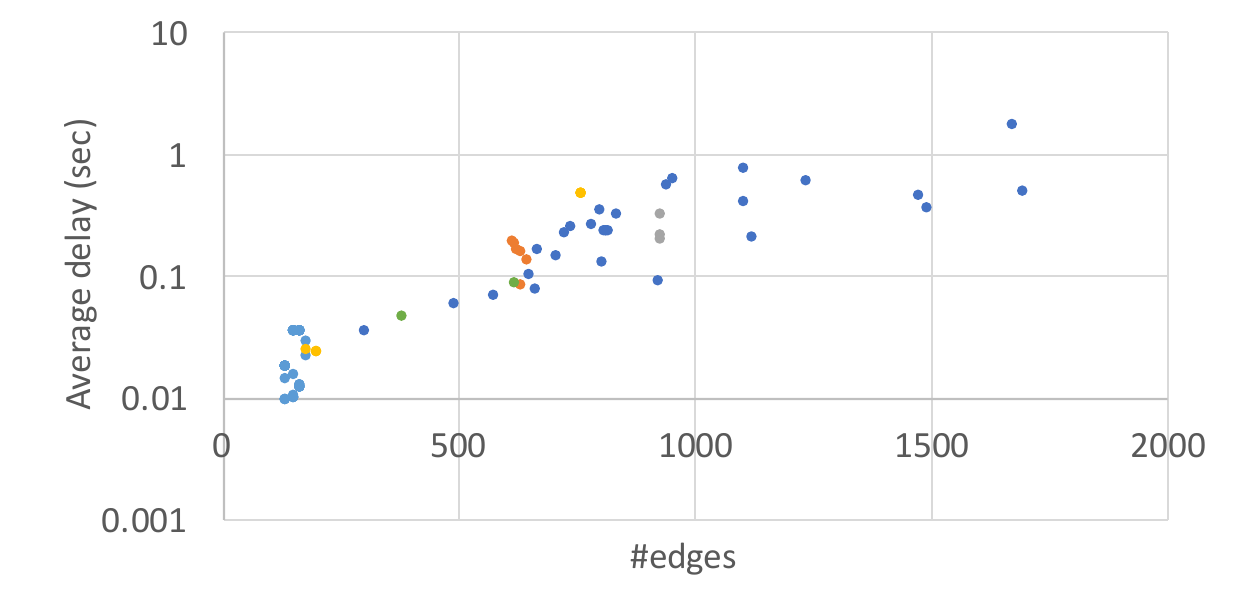}
	\caption{$\MSCM$}\label{fig:BNAvgDelayMCS_M}} 
	\end{subfigure}
	\caption{Average delay (in seconds) for the two triangulation algorithms
          over the probabilistic-graphical-model benchmarks: 
Object Detection ({\color{ObjectDetection}{$\bullet$}}),
Segmentation ({\color{Segmentation}{$\bullet$}}),
Pedigree ({\color{Pedigree}{$\bullet$}}),
 Grids ({\color{Grids}{$\bullet$}}),
{Promedas} ({\color{Promedas}{$\bullet$}}),
{CSP} ({\color{CSP}{$\bullet$}})
}
\end{figure*}

\subsubsection{Random graphs}
We measured the average delay (in seconds) between the printout of consecutive minimal
triangulations during a 30 minute execution.  The plots in
Figures~\ref{fig:RandomAvgDelay30minLB_TRIANG}
and~\ref{fig:RandomAvgDelay30minMCS} show the average delay vs.~the
size of the graph for the two variants.  We can see that the delay
increases with the size of the graph, and that the general trend is that the
delay is larger for denser graphs.  We also see that for $\LBTRI$ the
delay is generally longer than for $\MSCM$.

\begin{figure*}[tbh]
	\label{fig:RandomDelay}
	\centering
	\begin{subfigure}[b]{0.9\linewidth}{\includegraphics[width=0.9\textwidth]{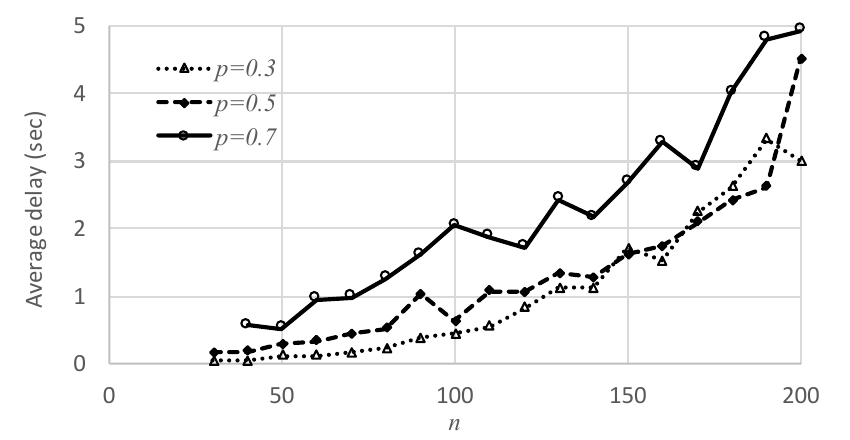}
	\caption{$\LBTRI$} \label{fig:RandomAvgDelay30minLB_TRIANG}}
	\end{subfigure}
	\begin{subfigure}[b]{0.9\linewidth}
	{\includegraphics[width=0.9\textwidth]{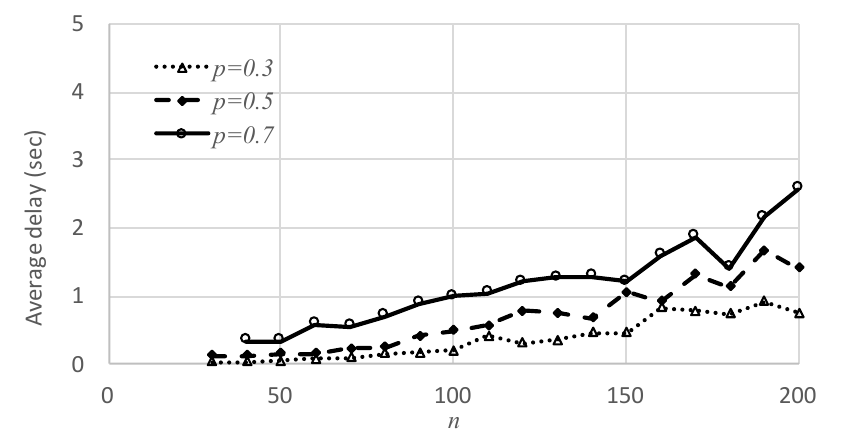} 
	\caption{$\MSCM$} \label{fig:RandomAvgDelay30minMCS}}
	\end{subfigure} 
	\caption{Average delay over 54 graphs randomly generated from
		the Erd\H{o}s-R\'enyi $G(n,p)$ for varying $n$ and $p$}
\end{figure*}

\subsubsection{Database queries}
We evaluated our enumeration algorithm over a set of $22$ queries from
the TPC-H dataset.  The graphs of these queries are quite small when
compared to the UAI datasets ($<23$ nodes).  Moreover, half of these
graphs are chordal to begin with (i.e., have only one minimal
triangulation---the graph itself), and hence, irrelevant for us.
Except for two queries, all of the rest had at most $5$ minimal
triangulations. The remaining two queries are \textsf{Q7} (Volume
shipping Query) and \textsf{Q9} (Product Type Profit Measure Query),
and they have a considerable number of minimal triangulations: $700$
and $588$, respectively.  When considering the minimum-width tree
decomposition for each of the queries, the largest bag was of size
$8$; this is due to a relation of arity $8$ in the query. In fact the
largest bag in each of the queries had at most two variables more than
the size of the largest relation.  The execution for all $22$ queries
completed within $5$ seconds.

In one of the queries we compared the delays for two modes of
printing: the one of $\algname{EnumMIS}$ and the one of
$\algname{EnumMISHold}$ that prints upon extraction from the queue, as described in
Section~\ref{sec:sgr-correct}. We refer to the former as \textsf{UG}
(Upon Generation) and to the latter as \textsf{UP} (Upon Pop).  Recall
that both modes guarantee incremental polynomial time
(Theorem~\ref{thm:inc-implies-inc}). 
\addedn{This gives us a sense of the practical impact of printing the solutions as soon as possible compared to holding the solutions to attain an easy-to-prove incremental polynomial time algorithm.}
The results are in
Figure~\ref{fig:TPCHDelay}. While the dotted line (of \textsf{UG}) has
bursts of high-frequency prints followed by periods where no new
triangulation is created, the solid line (\textsf{UP}) has a more
steady pace as can be seen by the constant slope in
Figure~\ref{fig:TPCHDelay}. As expected, despite the fact that the
last result of \textsf{UG} is printed earlier that of \textsf{UP},
termination is at the same time in both modes, as the algorithm still
needs to check that there are no additional minimal triangulations.

\begin{figure}[tbh]	
	\centering
	\includegraphics[width=0.9\textwidth]{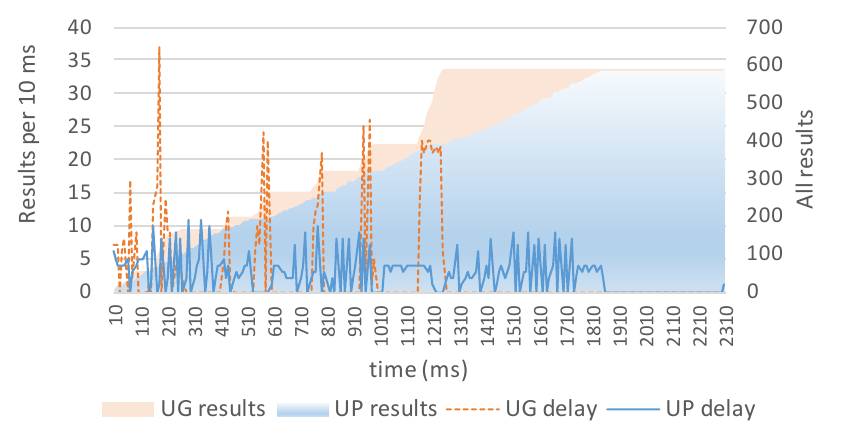} 
	\caption{Delay behavior in two printing modes:
		\textsf{UG} (Upon Generation, as in
		$\algname{EnumMIS}$), and \textsf{UP} (Upon Pop, as in
		$\algname{EnumMISHold}$)}
	\label{fig:TPCHDelay}
\end{figure}

\subsection{Quality}\label{sec:Quality}

In what follows we report on the quality of the generated minimal
triangulations in terms of two standard measures of quality for
triangulations and tree decompositions: \e{fill} and
\e{width}. \e{Fill} refers to the total number of edges added in order
to make the resulting graph chordal, while \e{width} refers to the
size of the largest clique in the generated
triangulation (minus one).\footnote{Recall that is NP-hard to find a triangulation
  that minimizes the fill~\cite{doi:10.1137/0602010} or the
  width~\cite{arnborg1987complexity}.}
The natural benchmark for the quality of the triangulations is the first result our enumeration returns, as it is the result we would get by running the minimal triangulation algorithm we used, on the original
input graph.

For each graph of the probabilistic inference dataset, we executed the enumeration algorithm for 30 minutes. The results in Table~\ref{tab:width} include only the experiments where the enumeration did not complete.
For each graph we measured the following: 
the number of generated triangulations (\textbf{\#trng}), 
the minimum observed width over all printed triangulations
(\textbf{min-w}), 
the number of printed triangulations of width at most that of the first
(\textbf{$\#{\leq}\textbf{w}_1$}), 
the average reduction in width (over the dataset) and the maximum
improvement in parentheses (\textbf{$\textbf{w}{\downarrow}$ $(\%)$}).
In Table~\ref{tab:fill} we show the same results for fill instead of width (\textbf{min-f},
\textbf{$\#{\leq}\textbf{f}_1$} and \textbf{$\textbf{f}{\downarrow}$
  $(\%)$}).

{
{
	\begin{table}[p]\small
		\renewcommand{\arraystretch}{1.2}
		\begin{center}	
			\begin{tabular}{|c|c|c|c|c|}
				\hline
				\textbf{Dataset} & 
				\textbf{$\#$trng} & \textbf{min-w} & \textbf{$\#{\leq}\textbf{w}_1$ $(\%)$} & \textbf{$\%\textbf{w}{\downarrow}$ $(\max)$} 
				\\
				\hline\hline
				\multicolumn{5}{|c|}{$\MSCM$}\\\hline
				Promedas $(28)$& $11064.5$ & $25.8$ & $3713.6$ ($33.6\%$) &  $2.2$ ($15.2$) \\
				\hline
				Grids $(8)$ & $40319.8$ & $18.4$ & $93.6~(0.2\%)$ &  $0.0$ ($0.0$) \\
				\hline
				Obj.~Detection $(79)$ & $100349.9$ & $6.1$ & $42743.9$ ($42.6\%$)& $0.4$ ($12.5$)\\
				\hline                  
				Segmentation $(5)$& $12836.5$ & $23.0$ & $20.5~(0.2\%)$  & $0.0$ ($0.0$) \\
				\hline
				Pedigree $(3)$& $7789.0$ & $31.7$ & $3087.3~(39.6\%)$ &$0.0$ ($0.0$)\\
				\hline
				CSP $(2)$& $29450.5$ & $16.5$ & $26741.5$ ($90.8\%$) & $13.2$ ($26.3$)\\
				\hline
				\multicolumn{5}{|c|}{$\LBTRI$}\\\hline
				Promedas $(28)$& $4220.7$ & $18.6$ & $2352.0$ ($55.7\%$) & $1.9$ ($16.7$)\\
				\hline
				Grids $(8)$ & $13881.3$ & $24.5$ & $1273.0$ ($9.2\%$) & $3.0$ ($8.7$) \\
				\hline
				Obj.~Detection $(79)$ & $33295.4$ & $5.8$ & $15709.3~(47.2\%)$ & $0.0$ ($0.0$)\\
				\hline
				Segmentation $(5)$& $5174.2$ & $21.8$ & $2141.8$ ($41.4\%$) & $10.3$ ($20.7$) \\
				\hline
				Pedigree $(3)$& $3646.0$ & $23.7$ & $3227.7$ ($88.5\%$) & $5.3$ ($14.8$) \\
				\hline
				CSP $(2)$& $11772.0$ & $16.5$ & $3760.5~(31.9\%)$ & $0.0$ ($0.0$)\\
				\hline
			\end{tabular}		
		\end{center}	
		\caption{Width statistics on generated triangulations following 30 minutes execution\label{tab:width}}
	\end{table}
}

{
	\begin{table}[p]\small
		\renewcommand{\arraystretch}{1.2}
		\begin{center}	
			\begin{tabular}{|c|c|c|c|c|}
				\hline
				\textbf{Dataset} & 
				\textbf{$\#$trng} & \textbf{min-f} & \textbf{$\#{\leq}\textbf{f}_1$ $(\%)$}& \textbf{$\%\textbf{f}{\downarrow}$ $(\max)$}
				\\
				\hline\hline
				\multicolumn{5}{|c|}{$\MSCM$}\\\hline
				Promedas $(28)$& $11064.5$ & $3353.4$ & $8136.0~(73.5\%)$ & $18.1$ ($49.9$) \\
				\hline
				Grids $(8)$ & $40319.8$ & $2752.6$ & $15771.4~(39.1\%)$ & $4.2$ ($28.1$)\\
				\hline
				Obj.~Detection $(79)$ & $100349.9$ & $30.0$ & $27614.1~(27.5\%)$  & $19.9$ ($47.1$)\\
				\hline                  
				Segmentation $(5)$& $12836.5$ & $2555.2$ & $5269.7~(41.1\%)$& $5.9$ ($12.5$) \\
				\hline
				Pedigree $(3)$& $7789.0$ & $3525.7$ & $743.0~(9.5\%)$ & $2.8$ ($3.5$) \\
				\hline
				CSP $(2)$& $29450.5$ & $46.0$ & $18815.5~(63.9\%)$  &  $35.2$ ($55.8$)\\
				\hline
				\multicolumn{5}{|c|}{$\LBTRI$}\\\hline
				Promedas $(28)$& $4220.7$ & $1239.4$ & $175.0~(4.1\%)$  &  $0.2$ ($11.1$)\\
				\hline
				Grids $(8)$ & $13881.3$ & $1600.3$ & $1.0~(0.0\%)$ & $0.0$ ($0.0$) \\
				\hline
				Obj.~Detection $(79)$ & $33295.4$ & $27.6$ & $5110.7~(15.3\%)$  & $10.4$ ($27.6$)\\
				\hline
				Segmentation $(5)$& $5174.2$ & $1402.0$ & $130.2~(2.5\%)$ & $1.2$ ($4.2$) \\
				\hline
				Pedigree $(3)$& $3646.0$ & $1491.0$ & $1.0~(0.0\%)$ & $0.0$ ($0.0$) \\
				\hline
				CSP $(2)$& $11772.0$ & $34.5$ & $664.0~(5.6\%)$ & $1.4$ ($3.0$) \\
				\hline
			\end{tabular}		
		\end{center}	
		\caption{Fill statistics on generated triangulations following 30 minutes execution\label{tab:fill}}
	\end{table}
}}

We can see that the algorithm, in both variants, is able to generate a
significant number of triangulations of high quality, in terms of both
width and fill. Moreover, it amplifies the quality of the underlying
triangulation, by means of width, and much more by means of fill.
According to the number of triangulations printed, $\MSCM$ enables
generating twice as many triangulations as $\LBTRI$.  However, with the exception of only a handful of the graphs tested, the
triangulations generated by $\LBTRI$ are superior in both the width
and fill metrics (this is especially apparent for the Promedas and
Pedigree datasets). Furthermore, this set of superior triangulations
accounts for a larger portion of the total number of triangulations
that generated.

\begin{figure}[t]
  \centering \includegraphics[width=0.9\textwidth]{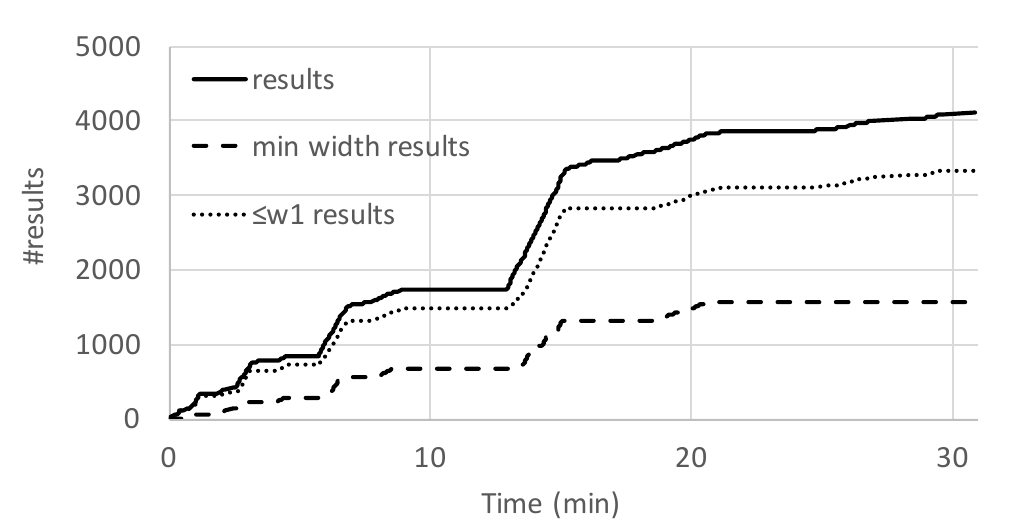} 
  \caption{Cumulative number of
    triangulations\label{fig:SingleExecutionCumulitiveResult}}
\end{figure}

\subsection{Case Study}\label{sec:testcase}
In this section we take a closer look at the behaviour of the
enumeration during a single execution. We use a graph from the
Promedas dataset. In Figure~\ref{fig:SingleExecutionCumulitiveResult}
we show the cumulative number of results generated over time. We
consider three types of results: \e{(a)} all minimal triangulations,
\e{(b)} minimal triangulations of the minimum width (where the minimum
is taken over the printed triangulations), \e{and (c)} those with a
width at most that of the first triangulation (${\leq}$w1), which is
the one that we would obtain by using only the triangulation algorithm
at hand. The reduction in the number of new triangulations over time
is consistent with the increase in the delay entailed in the guarantee
of incremental polynomial time, rather than polynomial delay.

Figure~\ref{fig:SingleExecutionMinFillWidth} presents the reduction in
the minimum width and minimum fill obtained during the execution of
the algorithm.  Each time slice records the minimum width (solid
curve) and minimum fill (dotted curve) observed up to that time. We
can see that both the width and the fill reduce over time, but the
minimum observed width is reached very quickly, while attaining the
minimum observed fill takes longer.

\begin{figure}
	\centering
	\includegraphics[width=0.9\textwidth]{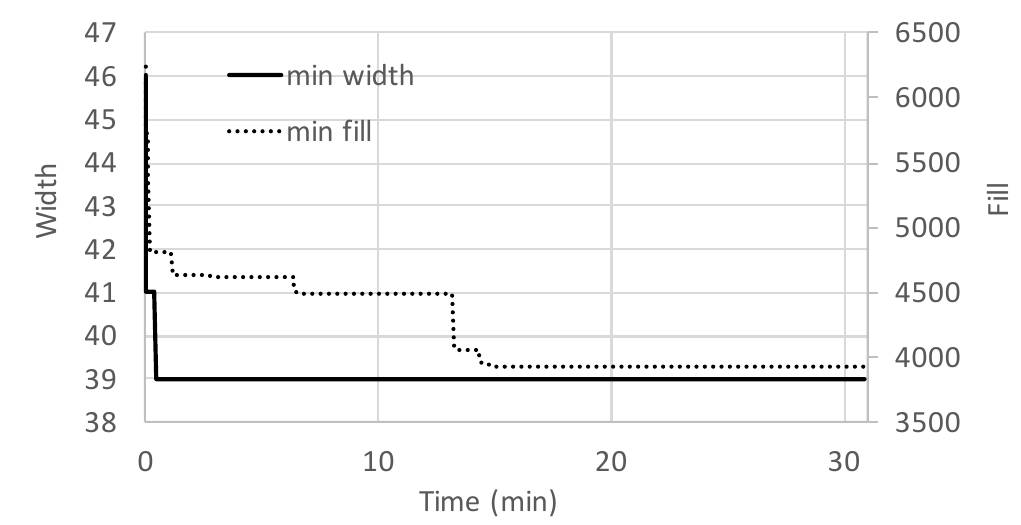}
	\caption{Minimum width and fill over time\label{fig:SingleExecutionMinFillWidth}}        
\end{figure}

\section{Concluding Remarks}\label{sec:conclusions}
We introduced the concept of a succinct graph representation (SGR),
and presented an enumeration algorithm for its maximal independent
sets. The algorithm enumerates in incremental polynomial time under
complexity assumptions: the SGR is tractably accessible, and it has a
tractable expansion. Consequently, we established an algorithm for
enumerating the minimal triangulations of a graph by reducing the
problem to the enumeration of the maximal independent sets of an SGR,
and showing that the complexity assumptions hold. We also proved that
enumerating the minimal triangulations enables the enumeration of the
proper tree decompositions. Our experimental study showed that the
algorithm is effective on graphs of various domains, and is able to
enhance off-the-shelf algorithms for triangulation (or tree
decomposition) by generating many (rather than just one) high-quality
different triangulations, and even improve standard quality measures
such as width and fill.

This work opens up quite a few directions for future work. On the
theoretical side, it is left open whether the enumeration of the
minimal triangulations can be carried out with polynomial delay. As discussed in Section~\ref{sec:poly} and Section~\ref{sec:sgrpdelay}, it is not clear how to do this with known techniques, and the abstraction used here cannot achieve this time bound.
Polynomial delay is possible in the case that the number of minimal separators of the input graph is polynomial in the size of the input graph.
In a follow up work to the originial publication of this manuscript, Ravid et al.~\cite{Ravid:2019} showed how to perform in such cases ranked enumeration under a wide class of cost functions that generalizes width and fill-in.
If the number of minimal separations is not bounded, the question of incorporating some order remains open.
\addedn{In terms of the space complexity, it is left open to determine whether the problem can be solved using polynomial (or even sub-exponential) space; the algorithm presented here uses exponential space in the worst case as it stores the minimal separators of the graph and the generated results.}
From the practical aspect, the algorithm presented here holds many opportunities for optimization over real-life graphs. An optimized version of the code is available online.\footnote{\url{https://github.com/TechnionTDK/efficient-td-enum}}

\subsection*{Acknowledgements}
This work was supported in part by the US-Israel Binational
Science Foundation (BSF) Grant No.~2014391, the Israel Science
Foundation (ISF) Grant No.~1295/15, and the Austrian Science Fund (FWF):
P30930-N35, W1255-N23. Benny Kimelfeld is a Taub Fellow,
supported by the Taub Foundation.
The authors are greatly thankful to LogicBlox, and in particular Hung
Ngo, for insightful discussions and for providing test
data.
 
\bibliography{main}
\bibliographystyle{abbrv}

\end{document}